%% file: PLambdaLONG.tex
\documentclass[conference,a4paper]{IEEEtran}

\RequirePackage{ifthen}

\newcommand{\version}{1}
\newcommand{\comment}{0} 

\newcommand{\condinc}[2]{\ifthenelse{\equal{\comment}{0}}{#1}{\green {**[[#2]]**}}}

\newcommand{\SLV}[2]{\ifthenelse{\equal{\version}{0}}{#1}{ #2}}

\input{macros_long}

\usepackage{floatflt}
\usepackage[font=scriptsize,skip=4pt]{caption}

\title{Lambda Calculus and Probabilistic Computation\\ (Extended Version)}
\author{
	\IEEEauthorblockN{Claudia Faggian}
	\IEEEauthorblockA{Universit\'e de Paris, IRIF, CNRS, France
	}
	\and
	\IEEEauthorblockN{Simona Ronchi della Rocca}
	\IEEEauthorblockA{Dip. di Informatica, Universit\`a di Torino, Italy
	}
}

\date{}

\begin{document}
\maketitle

\begin{abstract}
We introduce two extensions of the $\lambda$-calculus with a probabilistic choice operator, $\PLambda^{\cbv}$ and $\PLambda^{\cbn}$, modeling respectively  call-by-value and call-by-name probabilistic computation. We prove that both enjoys confluence and standardization, in an extended way:  we revisit these two fundamental notions  to take into account the asymptotic behaviour of terms. The common root of the two calculi is a further calculus  based on Linear Logic, $\PLambda^!$, which allows us to develop a unified, modular approach.	

%
%
\end{abstract}

\section{Introduction}
The pervasive role of  stochastic models   in a variety of domains (such as machine learning, natural language, verification) has prompted a vast  body of research on  probabilistic programming languages;  such a  language supports at least discrete distributions by  providing  an operator which models sampling. 
 In particular, the  functional style of probabilistic programming, pioneered  by   \cite{Saheb-Djahromi78},  attracts increasing  interest  because  it allows for  higher-order computation, and offers  a level of abstraction well-suited to deal with mathematical objects. Early work  \cite{KollerMP97, PlessL01, Park03, RamseyP02, ParkPT05}
 has evolved in a  growing  body of software development and  theoretical research. 
 In this context, the $\lambda$-calculus has often been used as a core language.  

%

In order to model higher-order probabilistic computation, it is a natural approach to take the $\lambda$-calculus as general paradigm, and to enrich it with a probabilistic construct.  The most simple  and concrete way to do so (\cite{DiPierroHW05,LagoZ12,EhrhardPT11}) is  to equip the untyped $\lambda$-calculus with an operator $\oplus$, 
which models  flipping a fair coin. 
This suffices to have universality, as proved in \cite{LagoZ12}, in the sense that the   calculus is 
	sound and complete with respect to   \emph{computable} probability distributions.
The resulting calculus is however \emph{non-confluent}, as it has been observed early (see \cite{LagoZ12} for an analysis).
 We revise the issue in  Example~\ref{ex:motivation}.   The  problem  with confluence is  handled in the literature by fixing a deterministic reduction strategy, typically  the leftmost-outermost strategy.
This is not satisfactory both for  theoretical and practical reasons, as we discuss later. 

In this paper, we  propose a more general point of view. Our goal is a foundational calculus, which plays   the same role as the $\lambda$-calculus does for deterministic computation. More precisely, taking the point of view propounded by  Plotkin in \cite{PlotkinCbV}, we  discriminate   between a \emph{calculus} and a  \emph{programming language}. The former defines the reduction rules, independently from any reduction strategy, and enjoys  confluence and standardization, the latter is specified by a deterministic strategy (an abstract machine). \emph{Standardization} is what relates the two: the programming language implements the standard strategy associated to the calculus. Indeed, standardization implies the existence of a strategy (the standard strategy) which is guaranteed  to reach the result, if it exists.


In  this spirit, we consider a probabilistic  \emph{calculus} to be characterized by a specific calling mechanism; 
the reduction is otherwise only constrained  by  the need of discriminating between duplicating  a function which samples from a distribution,  and duplicating the result of sampling. Think of tossing a coin and
duplicating the result, versus  tossing the coin twice, which is indeed the issue at the core of confluence failure, as the following examples  (adapted from \cite{deLiguoroP95,LagoZ12}) show.

\begin{example}[Confluence]\label{ex:motivation} 	
	Let us  consider  the  untyped $\lambda$-calculus extended with a binary operator $\oplus$ which models fair, binary probabilistic choice: 
	$M\oplus N$  reduces to either $M$ or $N$ with equal probability $1/2$; we write this as $M\oplus N \red \{M^{\frac{1}{2}},N^{\frac{1}{2}}\}$. Intuitively,  \emph{the result} of evaluating a probabilistic term is a \emph{distribution} on its possible values.
	
	\begin{enumerate}	
		\item 	Consider the term $PQ$, where $P=(\lam z.z \xor z)$, and  $Q=(\true\oplus \false)$;  $\xor$  is    the  standard construct for exclusive $\mathtt{OR}$,    
		$\true=\lam xy.x$ and $\false=\lam xy.y$  code  the  boolean values. \\
			-- If we first reduce  $Q$, 
			we obtain   $(\lam z.z \xor z)\true$ or $(\lam z.z \xor z)\false$, with equal probability $1/2$.		
			This way, $PQ$ evaluates to $\{\false^1\}$, \ie\ $\false$ with probability $1$.\\
			-- If we reduce  the outermost redex first, $PQ$  reduces to 
			$(\true\oplus \false) \xor (\true \oplus \false)$, and   the term evaluates to the distribution $\{\true^{\two}, \false^{\two}\}$. \\
	The two resulting distributions are not even comparable.

		\item 	The same  phenomenon appears even  if we restrict ourselves to 
		call-by-value.  Consider for example the reductions of $PN$ with $P$ as in 1), and $ N=(\lam xy.x\oplus y)$. We obtain the same two different  distributions as above.
		
		\condinc{}
		{$PQ\red \mset{Q \xor Q} \Red^* \mset{1/2 \lam xy.x, 1/2 \lam xy.y}$ ed anche  $PQ\red\mset{\two P(\lam xy.x),\two P(\lam xy.y)}\Red^*\mset{1 \lam xy.y}$   Mi sembra che il constraint e' che la riduzione probabilista sia surface.}
		
	\condinc{}{	\item 	The situation  becomes even more complex if we consider  the possibility of diverging; try the same experiment on the term $PR$, where $R=(\true\oplus \false)\oplus \Delta\Delta$ (for $\Delta=\lam x.xx$).}
	\end{enumerate}

\end{example}

In this paper, we define two probabilistic  $\lambda$-calculi,  respectively based on the call-by-value (CbV) and call-by-name (CbN) calling mechanism. Both enjoy confluence and standardization, \emph{in an extended way}: indeed we revisit these two fundamental notions to take into account the asymptotic behaviour of terms. The common root of the two calculi is a further calculus based on Linear Logic, which is an extension of  Simpson's linear $\lambda$-calculus \cite{Simpson05}, and which allows us to develop a unified, modular approach.\\

\paragraph*{Content and Contributions}

In  Section~\ref{sec:cbv}, we introduce  a call-by-value   calculus, denoted  $\PLambda^\cbv$, as a probabilistic extension of the call-by-value $\lambda$-calculus of Plotkin (where the $\beta$-reduction fires only in case the argument is a value, \ie\, either a variable or a $\lambda$-abstraction). We choose to study in detail \emph{call-by-value} for two main reasons. 
First,  it is the most relevant mechanism  to probabilistic  programming (most of the abstract languages we cited are  call-by-value, but also real-world stochastic programs such as Church \cite{GoodmanMRBT08}). Second, call-by-value is a mechanism in which dealing with functions, and duplication of functions,  is clean and intuitive, which allows us to address   the issue at the core of confluence failure.
The definition of value (in particular, a probabilistic choice is not a value) together with a suitable restriction of the evaluation context for the probabilistic choice, allow us to recover  key results:   confluence and a form of standardization (Section~\ref{sec:conf_and_st}). Let us recall that, in the classical $\lambda$-calculus, standardization means  that there is a  strategy which is \emph{complete} for all reduction sequences, \ie, for every  reduction  sequence $M \red^*N$  there is a \emph{standard} reduction sequence from $M$ to  $N$. 
A standard reduction sequence with the same property exists also here. An unexpected result is that strategies which are complete in the classical case, are not so here,  notably the leftmost  strategy.

In Section~\ref{sec:asymptotic} we study the asymptotic behavior of terms. 
Our leading question is how  the asymptotic behaviour of different sequences starting from the same term compare. 
We first analyze if and  in which sense  confluence implies that  the result of a \emph{probabilistically terminating} computation is unique. We formalize the notion of \emph{asymptotic result} via   \emph{limit distributions}, and establish that there is a \emph{unique} maximal one. 

In Section~\ref{sec:left_evaluation} we address   the question of how to find 
such  greatest limit distribution, a question which arises from the fact that evaluation in $\PLambda^\cbv$ is  non-deterministic, and different sequences may terminate with different probability.
With this aim, we extend the notion of standardization to  limits; this extension is non-trivial, and  demands the development of new  sophisticated proof methods.

We prove that the new notion of standardization supplies a family of \emph{complete reduction strategies} which are guaranteed  to reach  the unique maximal result. 
	Remarkably, we are able to show that,  when evaluating  programs, i.e., closed terms, this family  \emph{does include} the leftmost strategy.  As we have already observed, this is the   deterministic strategy which  is typically adopted in the literature, in either  its call-by-value (\cite{KollerMP97,DalLagoMZ11}) or its  call-by-name version (\cite{DiPierroHW05,EhrhardPT11}), but without any completeness result with respect to \emph{probabilistic} computation. 
	  Our  result offers an ``a posteriori'' justification for its  use!

The study of $\PLambda^\cbv$ allows us to develop a crisp approach, which we are then able to use in the study of  different probabilistic  calculi.
Because the issue of duplication is central, it is natural to expect a benefit from  the fine control over copies which is provided by Linear Logic. In Section~\ref{sec:Plinear} we use our tools to introduce and study a probabilistic \emph{linear} $\lambda$-calculus, $\PLambda^!$.  The linear calculus  provides not only a finer  control on duplication, but also a modular approach to  confluence and standardization, which allow us to formalize a call-by-name version of our calculus, namely $\PLambda^{\cbn}$, in Section~\ref{sec:CBN}. We prove that $\PLambda^{\cbn}$ enjoys  properties analogous to those of  $\PLambda^\cbv$, in particular confluence and standardization.

In Section~\ref{sec:issues} we provide the reader with some    background and motivational observations. 
Basic notions of discrete probability and rewriting are reviewed in Section~\ref{sec:preliminaries}.\\

\SLV{An extended version of this paper, with more details and   proofs,   is available online \cite{long}.}{}

\paragraph*{Related Work}

The idea of extending the $\lambda$-calculus with a probabilistic construct  is  not new; 
without any ambition to be exhaustive, let us cite \cite{Park03,RamseyP02}, \cite{DiPierroHW05,EhrhardPT11,LagoZ12,BorgstromLGS16}. In all these cases, a specific reduction strategy is fixed; they are indeed languages, not calculi, according to   Plotkin's distinction.

The issue about confluence appears every time the $\lambda$-calculus is extended with a  choice effect: quantum, algebraic,  non-deterministic. The ways of framing the same problem in different settings are naturally related, and we were inspired by them.  Confluence for an algebric calculus  is dealt with in  \cite{ArrighiD17}  for the call-by-value, and in   \cite{Vaux09} for the call-by-name. In the quantum case we would like  to cite   \cite{DalLagoMZ11, popl17}, which are based on   Simpson's calculus \cite{Simpson05}. {A probabilistic extension of Simpson's calculus was first  proposed in \cite{DiazMartinez17}.
The language is similar to that of $\PLambda^!$; however in \cite{DiazMartinez17} (as also in \cite{DalLagoMZ11, popl17}) no reduction (not even $\beta$) is allowed in the scope of a !-operator. The reduction there hence corresponds to  surface reduction,  which in Sec.~\ref{sec:Plinear} we show to be the  \emph{standard strategy} for  $\PLambda^!$.}

To our knowledge, the only proposal of a probabilistic $\lambda$-calculus  in which the reduction is independent from a \emph{specific  strategy} is for  \emph{call-by-name}, namely the calculus of \cite{LeventisThesis},  in the line  of work of differential  \cite{EhrhardR03} and  algebric \cite{Vaux09} $ \lambda $-calculus.  
The focus in \cite{LeventisThesis}  is essentially semantical, as the author want to study an equational theory for the $\lambda$-calculus, based on an extension of B\"ohm trees. 
\cite{LeventisThesis} develops results which in their essence are similar to those we obtain for call-by-name in Sec.~\ref{sec:CBN}, in particular confluence and standardization, even if 
his calculus --which internalizes the probabilistic behavior-- is quite different from ours, and so are the proof techniques. 

Finally,
we wish   to mention that  proposals of a probabilistic $\lambda$-calculus  could also be extracted from semantical models, such as  
the one in 
\cite{BacciFKMPS18}, which develops an idea earlier presented in \cite{scott2014}, and in which the notion of graph models for $\lambda$-calculus has been extended with  a probabilistic construct.

\section{Background and Motivational observations}\label{sec:issues}

In this section, we  first review -in a non-technical way-  the specific features of probabilistic programs, and how they differ from classical ones. 
We then focus on some motivational observations which are relevant to  our work.  First, we give  an example of features which are lost if a programming language is characterized by a  strategy which is not rooted in a more general calculus.  Then, we  illustrate some of the  issues which appear when we study a general calculus, instead of a specific reduction strategy. Addressing  these issues  will lead us to  develop new notions and tools.

\subsection{Classical vs. Probabilistic Programs}

A \emph{classical} program
defines a \emph{deterministic input-output} relation; it 
terminates (on a given input), or  does not; if it terminates, the program only runs for  a finite number of steps.
Instead, a \emph{probabilistic} program   generates a probability \emph{distribution over possible outputs}; it
\emph{terminates}  (on a given input) \emph{with a certain probability};
it  may have runs which take \emph{infinitely many steps}	even when termination has probability $1$.

A probabilistic  program is a stochastic model. 
The intuition is that  the probabilistic  program $P$  is executed, and random choices are made by sampling; 
this process  defines a distribution  over all the possible outputs of $P$. 
Even  if  the   termination probability is $1$ (\emph{almost sure termination}), that degree of certitude 
is typically not reached in any finite number of steps, but it appears \emph{as a limit}.  A standard example is a term $M$ which reduces to either the normal form $\true$ or $M$ itself, with equal probability $1/2$. After $n$ steps, $M$ reduces to $\true$ with probability 
$\two + \frac{1}{2^2} + \dots + \frac{1}{2^n}$. 
\emph{Only at the limit} this computation terminates with probability $1$ .

\paragraph*{Probabilistic vs. Quantitative} 
The notion of \emph{probabilistic termination} is what sets apart   probabilistic $\lambda $-calculus from   other quantitative calculi such as those in \cite{ArrighiD17,EhrhardR03,Vaux09}, and from the non-deterministic $\lambda $-calculus \cite{deLiguoroP95}.
For this reason,  the \emph{asymptotic behaviour} of terms will be the focus of this paper. 

\subsection{Confluence of the calculus  is relevant to programming}\label{mot:confluence}
 
 Functional languages have their foundation in the $\lambda$-calculus and its properties, and such properties (notably, confluence and standardization) have theoretical and practical implications. 
A   strength of  \emph{classical} functional languages -which is assuming growing importance-
is that they  are inherently parallel (we refer e.g. to \cite{Marlow} for discussion on \emph{deterministic} parallel programming): every sub-expression can be evaluated in parallel,
because of referential transparency; still, we can perform  reasoning, testing and debugging on a program 	using a  sequential model, 
	because the result of the calculus is independent from the evaluation order.
Not to force a  sequential strategy   impacts  the implementation of the language, but also  the \emph{conception} of  programs. As    advocated  by Harper,  
 the  parallelism of functional languages exposes the 
{\textsl{
		``dependency structure of the computation by not introducing any dependencies that are not forced on us by the nature of the computation itself."}
}

This feature   of functional languages is rooted in the  \emph{confluence} of the $\lambda$-calculus, and is   an example of what   is lost in the probabilistic setting, if we give-up either confluence, or the possibility of non-deterministic evaluation.

\condinc{}{
	The following  example is adapted from  \cite{deLiguoroP95,LagoZ12}.
	\begin{example}[Motivating Example]\label{ex:motivation}
		Let us consider  the  untyped $\lambda$-calculus extended with a binary operator $\oplus$ which models fair, binary probabilistic choice: 
		$M\oplus N$  reduces to either $M$ or $N$ with equal probability $1/2$. 

		Consider the term $PQ$, where $P=(\lam x. x)(\lam x.x \xor x)$ and $Q=(T\oplus F)$ (with $\xor$  the  standard constructs for exclusive $\mathtt{OR}$, and  
		$T,F$ terms which code the boolean values). 
		\begin{itemize}
			\item If we evaluate $P$ and $Q$ independently, from $P$ we obtain $\lam x.x \xor x$, while from $Q$ we have either $T$ or $F$, with equal probability $1/2$. By composing the partial results, we obtain $F$ with probability $1$.
			\item if we evaluate $PQ$ sequentially, in a standard left-most outer-most fashion, $PQ$ reduces to $(\lam x.x \xor x)Q$ which reduces to 
			$(T\oplus F) \xor (T\oplus F)$ and eventually  to $T$ with probability  $1/2$  and $F$ with probability $1/2$.
		\end{itemize}
		The situation  becomes even more complex if we consider also the possibility of diverging; try the same experiment on the term $PR$, where $R=(T\oplus F)\oplus \Delta\Delta$ (with $\Delta=\lam x.xx$)!
	\end{example}
}

\subsection{The result of probabilistic computation}\label{mot:result}

A ground for  our approach is the distinction between calculus and language. 
Some of the issues which we will  need to address  do not appear when working with probabilistic languages, because they are based on a  simplification of the $\lambda$-calculus. 
Programming languages  only evaluate  programs, i.e., \emph{closed terms} (without free variables). A striking  simplification  appears from another crucial restriction,   \emph{weak evaluation}, which does not evaluate function bodies (the scope of  $\lambda$-abstractions). In  \emph{weak call-by-value} (base of  the ML/CAML family of  probabilistic languages) values are normal forms.

What is the result of a probabilistic computation is well understood only in the case of  \emph{programming languages}: 
the result of a program is a distribution on its  possible outcomes, which are  \emph{normal forms} w.r.t. a chosen strategy.
In the literature of probabilistic $ \lambda $-calculus,  two main deterministic strategies have been studied:  weak left strategy in CbV  \cite{LagoZ12} and head strategy in CbN \cite{EhrhardPT11}, whose normal forms are respectively the closed values and the head normal forms. 

When considering a calculus instead of a language, the identity between normal forms and results does not hold anymore, with important consequences in the definition of limit distributions.
We investigate this issue in Sec.~\ref{sec:asymptotic}. 
The approach  we develop is general and uniform  to all our calculi.

\section{Technical Preliminaries}\label{sec:preliminaries}

We  review  basic notions on discrete probability and rewriting which we use through the paper. 
 We assume that the reader has some familiarity  with  the $\lambda$-calculus.

\subsection{Basics on Discrete Probability}\label{sec:proba}

A \emph{discrete probability space} is given by a pair $(\Omega, \mu)$,
 where  $\Omega$  is  a \emph{countable} set, and $\mu$ is  a \emph{discrete probability distribution}
on  $\Omega$, \ie\  is a 
function from $\Omega$ to $[0,1]\subset  \Real$ such that  $\norm \mu:= \sum_{\omega\in \Omega} \mu(\omega) = 1$.  In this case, 
a probability measure is assigned  to any  subset  $\A\subseteq \Omega$ as $\mu(\A)=\sum_{\omega\in \A} \mu(\omega)$.
 In the language of probability theory, a subset of $\Omega$ is called an \emph{event}.

Let  $(\Omega,\mu)$ be as above. Any \emph{function} $F : \Omega \to \Delta$, where  $\Delta$ is another countable set, 
 \emph{induces a  probability distribution} $\mu^F$ on $\Delta$
by composition: $\mu^F(d \in \Delta) := \mu (F^{-1}(d))$ \ie\ $\mu\{\omega \in \Omega: F(\omega) = d\}$. 
In the language of probability theory, $F$ is called a \emph{discrete random variable} on $(\Omega, \mu)$.

\begin{example}[Die]\label{ex:odd} 
\begin{enumerate}
\item Consider tossing a die once. The  space of possible outcomes is the
	set $\Omega = \{1, 2, 3, 4, 5, 6\}$. The probability measure $\mu$ of each outcome is $1/6$. The event \emph{``result is odd"} is the subset $\mathcal{O}= \{1, 3, 5\}$, whose probability measure is $\mu(\mathcal{O})= 1/2$.
	\item Let $\Delta$ be a set with two elements $\{\texttt{Even}, \texttt{Odd}\}$, and $F$ the obvious function from $\Omega$ to 
	$\Delta$. $F$ induces a distribution on $\Delta$, with $\mu^F(\texttt{Even})=1/2$ and $\mu^F(\texttt{Odd})=1/2$.
	\end{enumerate} 
\end{example}


\subsection{Subdistributions and $\pmb{\DST{\Omega}}$}\label{notation:dist}

 Given a countable set $\Omega$,  a function $\mu:\Omega\to[0,1]$ is a probability \emph{subdistribution} if 
 $\norm \mu\leq 1$.
 We write  $\DST{\Omega}$ for the set
 of   subdistributions on $\Omega$.   With a slight abuse of language, we will use the term distribution  also for   subdistribution. 
  Subdistributions allow us to deal with partial results and non-successful computations.
  
\emph{Order:}   $\DST{\Omega}$ is equipped with the standard order relation of functions :  $\mu \leq \rho$ if
$\mu (\omega) \leq \rho (\omega)$ for each $\omega\in \Omega$.
%
%

\emph{Support:} The \emph{support} of  $\mu$ is $\supp{\mu}=\{\omega:\mu(\omega)>0\}$.

\emph{Representation:} 
 We represent a distribution by explicitly indicating the support, and (as superscript) the probability assigned to each element by $\mu$. We  write $\mu=\{a_0^{p_0}, \dots , a_n^{p_n}\}$ if $\mu(a_0)=p_0,\dots , \mu(a_n)=p_n$ and  $\mu(a_j)=0 $ otherwise. 

\subsection{Multidistributions}\label{sec:multi} 
\newcommand{\X}{\mathcal X}
To syntactically represent  the global evolution of a probabilistic system, we rely on the notion of multidistribution \cite{Avanzini}.

A \emph{multiset} is  a (finite) list of elements,  modulo reordering,  \ie\  $\mset{a,b,a}=\mset{a,a,b}\not=\mset{a,b}$; the multiset $\mset{a,a,b}$ has three elements. 
Let 
$\X$ be a  countable set  
and  $\m$  a multiset  of pairs of the form $pM$, with $p\in]0,1]$, and $M\in \X$.
 We call $\m=\mset{p_iM_i\st i\in I}$ (where  the  index set $I$ ranges over the elements of $\m$)   a \emph{multidistribution on $\X$}  if   
$ \sum_{\iI} p_i\leq 1$. We denote by $\MDST \X$  the set of all multidistributions on $\X$.

We  write the multidistribution $\mdist{1M}$ simply as $\mdist{M}$.
The sum of  multidistributions is denoted by $+$, and it is the concatenation of lists. 
The product $q\cdot \m$ of a scalar $q$ and a multidistribution $\m$ is defined pointwise: $q\cdot \multiset{ p_{1}M_{1},..., p_{n}M_{n}}=\multiset{ (qp_{1})M_{1},..., (qp_{n})M_{n}} $. 

Intuitively, a multidistribution $\m\in \MDST{\X} $ is a syntactical representation  of a discrete probability space   where at each element  of the space  is associated a probability and a term of $\X$.
To the multidistribution $\m=\mset{p_iM_i\st i\in I}$    we associate a  probability distribution $\mu \in \DST{\X}$ as follows: 
%
\[ \mu(M) = 
\left\{
\begin{array}{ll}
p& \mbox{ if } p=\sum_{\iI}p_i \mbox{ s.t. }   M_i=M\\
0&\mbox{otherwise;}
\end{array}
\right.
\]
and we call $\mu$ the \emph{probability distribution associated} to $\m$.

\begin{example}[Distribution vs. multidistribution]
	If $\m=\mset{\frac{1}{2}a,\frac{1}{2}a}$, then $\mu=\{a^1\}$. Please observe  the \emph{difference} between distribution and multidistribution: if $\m'=\mset{1a}$, then $\m\not= \m'$, but $\mu=\mu'$.
\end{example}

\newcommand{\R}{\mathcal X}
\subsection{Binary relations (notations and basic definitions)}\label{sec:rel}  
Let  $\red_r$ be  a binary relation on a set $\R$. We denote  
$\red_r^*$  its  reflexive and transitive closure.
We denote $=_r$ the reflexive, symmetric and transitive closure of $\red_r$.
 If $u\in \R$, we   write  $u\not\red_r $
  if there is no $t\in \R$ such that 
$u\red_r t$; in this case, $u$ is in $\red_r$-\emph{normal form}.
\textbf{Figures convention:} as is standard, in the \emph{figures} we depict $\rightarrow^*$  as $\twoheadrightarrow$; solid arrows are universally quantified, dashed arrows are existentially quantified.

\paragraph*{Confluence and Commutation} Let  
$r,s,t,u\in {\R}$. 	The relations $\red_{1}$ and $\red_{2}$  on $\R$
\emph{commute} if  ($r \red^{*}_{1}s$ and $r\red^{*}_{2}t$) imply there is $u$ such that 
($s \red^{*}_{2}u$ and $r_{3}\red^{*}_{1}u$); they 
\emph{diamond-commute} ($\diamond$-commute) if ($r\red_{1}s$ and $r\red_{2}t$) imply there is $u$ such that ($s \red_{2}u$ and $t\red_{1}u$). The relation 
$\red$ is \emph{confluent} (resp. \emph{diamond}) if it commutes (resp. $\diamond$-commutes) with itself. It is well known that 
$ \diamond $-commutation implies commutation, and diamond implies confluence.

\section{Call-by-Value calculus $\PLambda^\cbv$}\label{sec:cbv}
We define $\PLambda^\cbv$, a CbV probabilistic $\lambda$-calculus.

\condinc{}{
\CF{Notare: CbV di per se non garantisce confluenza.
	 Esempio. Let $P:=(\lam z. z \xor z)$ and $ Q:=(\lam xy.x\oplus y)$. Con chiusura per contesto generale $\cc$, la riduzione di $PQ$ non e' confluente. 
 $PQ\red \mset{Q \xor Q} \Red^* \mset{1/2 \lam xy.x, 1/2 \lam xy.y}$ ed anche  $PQ\red\mset{\two P(\lam xy.x),\two P(\lam xy.y)}\Red^*\mset{1 \lam xy.y}$   Mi sembra che il constraint e' che la riduzione probabilista sia surface.}

}

\subsection{Syntax of $\PLambda^\cbv$}\label{sec:syntax}
\subsubsection{The language}\label{def:language}
\emph{Terms} and \emph{values} are generated respectively by the grammars:
\begin{center}
{\footnotesize {\begin{minipage}[c]{0.45\textwidth}
		$
		\begin{array}{lcllr} M,N,P,Q & ::= & x \mid \lambda x.M \mid MM \mid M \oplus M & (\textbf{terms } \PLambda)\\
		V,W & ::= & x \mid \lambda x. M & (\textbf{values } \Val)\\
	
		\end{array}
		$
\end{minipage}}}
\end{center}
where $x$ ranges over a countable set of \emph{variables} (denoted by $x, y,  \dots$).
$\Lambda_\oplus$ and $\Val$ denote respectively the set of terms and of values.
Free variables are defined as usual. $M[N/x]$ denotes the term obtained by capture-avoiding substitution of $N$ for each free occurrence of $x$ in $M$.

 \emph{Contexts} ($\cc$) and \emph{surface contexts} ($\ss$) are generated  by the grammars:
 \begin{center}
{\footnotesize  {\begin{minipage}[c]{0.48\textwidth}
		$
		\begin{array}{lcllr} 
			\ccontext & ::= & \square   \mid M\ccontext \mid \ccontext M \mid \lambda x.\ccontext \mid \ccontext \oplus M \mid
		M \oplus \ccontext  & (\textbf{contexts})\\
		\ss & ::=& \square \mid M\ss \mid \ss M  & (\textbf{surface contexts})
		\end{array}
		$
\end{minipage}}}
 \end{center}
where $\square$ denotes the \emph{hole} of the term context. 
Given a term context $\ccontext$, we denote by $\ccontext(M)$ the term obtained from $\ccontext$ by filling the hole with $M$, allowing the capture of free variables.
 All surface contexts are  contexts.  Since the  hole will be filled with a redex, \emph{surface contexts}  formalize the fact that the redex (the hole) is not in the scope of a $\lambda$-abstraction, nor of a $\oplus$.

$\MDST{\Lambda_{\oplus}}$ denotes the set of  \emph{multi-distributions} on $\Lambda_\oplus$.

\subsubsection{Reductions} We first define reduction rules on terms (Fig.~\ref{fig:rules}),  and  one-step reduction from terms to multidistributions (Fig.~\ref{fig:reductions}). We  then lift the definition of reduction to a binary relation on $\MDST{\PLambda}$.

Observe that, usually, a   reduction step is 
given by the closure under context of the reduction rules. 
  However, to define a  reduction from term to term is not informative enough, because we still have to account for the probability. 
The  meaning of $M\oplus N$ is that this term reduces to
either  $M$ or $N$, \emph{with equal probability} $\frac{1}{2}$. There are various way to formalize this fact; here, we use multidistributions.

\paragraph{Reduction Rules and Steps}
The \emph{reduction rules} on the terms of $\PLambda$ are defined  in Fig.~\ref{fig:rules}.\\
%
%
%

\vspace*{-8pt}
	\begin{figure}[!h]\centering
		{\scriptsize 	
		\begin{minipage}{0.5\textwidth}
		\begin{tabular}{|c|c|} 
			\hline
		\emph{$\beta_{v}$-rule} &\emph{ Probabilistic rules}\\
		$	\begin{array}{l }
		(\lambda x.M)V \mapsto_{\beta_v} M[V/x] ~ \mbox{ if } V \in \Val
		\end{array}$
		& 
		$ \begin{array}{l }  
		M \oplus N \mapsto_{l\oplus} M  \quad 
		M \oplus N \mapsto_{r\oplus} N 
		\end{array}$\\		
		\hline			
	\end{tabular} 
	  	\caption{Reduction Rules }\label{fig:rules}
		\end{minipage}
	}\end{figure}
\vspace*{-4pt}
 The \emph{(one-step) reduction relations}
$\redbv,\redo\subseteq \PLambda \times \MDST{\PLambda}$   are defined in Fig.~\ref{fig:reductions}.  Observe that the probabilistic rules $\mapsto_{r\oplus,l\oplus}$ are  closed only under surface contexts, while  the reduction rule $\mapsto_{\beta_v}$ is closed under general context $\cc$ (hence $\PLambda^\cbv$ is a conservative extension of Plotkin's  CbV $\lambda$-calculus, see \ref{sub:translation}).
We denote by $\red$ the union $\redbv \cup \redo$.\\
\vspace*{-8pt}
	\begin{figure}[!h]\centering
{\scriptsize 
	\fbox{
	\begin{minipage}[c]{0.45\textwidth}
\[
\infer{\ccontext((\lambda x.M)V) \redbv \multiset{\ccontext(M[V/x])}}{	(\lambda x.M)V \mapsto_{\beta_v} M[V/x] & V\in \Val}
\]

\[
\infer{\ss(M\oplus N) \red_{\oplus} \multiset{\frac{1}{2}\ss(M), \frac{1}{2}\ss(N)}}{M \oplus N\mapsto_{l\oplus} M\quad M \oplus N \mapsto_{r\oplus}N}
\]
	\end{minipage}
}}\caption{Reduction Steps }\label{fig:reductions}	\end{figure}
\vspace*{-4pt}

\condinc{}{
\frame{
	\begin{minipage}[c]{0.45\textwidth}
		Closure under context (multidistributions) :
\[
\infer{\ccontext(M) \redbv \multiset{1\ccontext(N)}}{\ccontext(M) \redbv \ccontext(N)}
\]
\[
\infer{\ccontext(M\oplus N) \red_{\oplus} \multiset{\frac{1}{2}\ccontext(M), \frac{1}{2}\ccontext(N)}}{\ccontext(M \oplus N) \red_l \ccontext(M)\quad \ccontext(M \oplus N) \red_r \ccontext(N)}
\]
	\end{minipage}
}
We denote by $\red$ the union $\redbv \cup \redo$.
}

\paragraph{Lifting}\label{def:lift}
We lift  the reduction relation $\red\subseteq \PLambda\times \MDST{\PLambda} $ to a  relation   $\Red\subseteq \MDST{\PLambda}\times \MDST{\PLambda}$, as defined in Fig.~\ref{fig:lifting}. Observe that $\Red$ is a reflexive relation.\\
\vspace*{-8pt}
	\begin{figure}[!h]\centering
	{\scriptsize 
		\frame{
			\begin{minipage}[c]{0.45\textwidth}
				\[
				\infer[L1]{\mset{M}\Red \mset{M}}{} \quad \quad
				\infer[L2]{\mset{M}\Red \m}{M\red\m}  \quad  \quad 
				\infer[L3]{ \multiset{p_{i}M_{i}\mid i\in I} \Red  \sum_{\iI} {p_i\cdot \m_i}} 
				{(\mset{M_i} \Red \m_i)_{\iI} }
				\]
			\end{minipage}	  	
}}\caption{Lifting of $\red$ }\label{fig:lifting}\end{figure}
\vspace*{-4pt}

We define in the same way the lifting of any relation   ${\red}_{r}\subseteq \Lambda_\oplus \times \MDST{\Lambda_\oplus}$ to a binary relation ${\Red}_{r}$ on  $\MDST{\PLambda}$. In particular, we lift  $\redbv,\redo$ to  $\Redbv,\Redo$.

\condinc{}{\begin{remark*}We observe that the set of  rules above which define lifting can also be  stated in a more compact way as follows:	
	\begin{center}
			{\footnotesize $
		\infer[\texttt{Lift}]{\s+ \multiset{p_{j}M_{j}\mid \jJ} \Red \s + \sum_{\jJ} {p_j\cdot \m_j}}
		{(M_j\red \m_j)_{\jJ} }
		$ }
	\end{center}
\end{remark*}
}

	\paragraph{Reduction sequences}
	A $\Red$-sequence (\emph{reduction sequence}) from $\m$ is a 
	sequence $\m=\m_0, \dots, \m_i, \m_{i+1}, \dots$  such that $\m_{i} \Red\m_{i+1}$ ($\forall i$). 
	We write  $\m\Red^*\n$ to indicate that there is a \emph{finite}  sequence from $\m$ to $\n$, and $\seq \m$ 
	 for  an  \emph{infinite   sequence}.
	

\paragraph{$\beta_{v}$ equivalence}\label{eqbeta}
We write 
$\equbetav$ for  the transitive, reflexive and symmetric closure  of $\Redbv$;  abusing the notation, we will 
write   $M\equbetav N$ for $\mset{M} =_{\beta_v}\mset{N} $.

\paragraph{Normal Forms}\label{def:nf} $\Nnf$ denotes the set of $\red$-normal forms.
Given  $\red_r\in \PLambda\times\MDST{\PLambda}$, 
a term $M$ is in  $\red_r$-\emph{normal form} if $M\not\red_r$, \ie\ there is no $\m$ such that $M\red_r \m$.   
It is easy  to check that all closed $\red$-normal forms are values, however a value is not necessarily a $\red$-normal form.

\subsubsection{Full Lifting}\label{def:fulllift}	 The definition of lifting allows us  to apply a reduction step $\red$ to any number of  $M_i$ in the multidistribution $\m=\mset{p_iM_i\st \iI}$. 
If no $M_i$ is reduced, then $\m\Red\m$ (the relation $\Red$ is  reflexive).
Another important case is when \emph{all} $M_i$ for which a reduction step is possible are indeed reduced. 
%
This notion of \emph{full} reduction, denoted by $\full$, is defined as follows.
\begin{center}
	{\footnotesize 	
	$
	\infer {\mset{M}\full\mset{M}}{M\not \red} \quad \quad
	\infer{\mset{M}\full \m}{M\red\m}  \quad  \quad 
	\infer{ \multiset{p_{i}M_{i}\mid i\in I} \full \sum_{\iI} {p_i\cdot \m_i}} 
	{(\mset{M_i} \full\m_i)_{\iI} }
	$ }
\end{center}
Obviously, $\full \subset \Red$. 
Similarly to lifting, also the notion of full lifting can be extended to any reduction. For any ${\red}_{r}\subseteq \Lambda_\oplus \times \MDST{\Lambda_\oplus}$, its full lifting is denoted by $\full_r\subseteq \MDST{\PLambda}\times\MDST{\PLambda}$.
 The relation $\full$  plays an important  role in \ref{sec:left_evaluation}.
 
\subsection{$\PLambda^\cbv$ and the $\lambda$-calculus}	\label{sub:translation}\label{sec:tr_cbv}
A comparison between $\PLambda^\cbv$ and the $\lambda$-calculus is in order. 

Let  $\Lambda$ be  the set  of  $\lambda$-terms; we denote 
by $\Lambda^\cbn$  the CbN $\lambda$-calculus, equipped with the reduction $\red_{\beta}$ \cite{Barendregt}, and by $\Lambda^\cbv$ the  CbV $\lambda$-calculus, equipped with the reduction $\redbv$ \cite{PlotkinCbV}.

$\PLambda^\cbv$  is a conservative extension of $\Lambda^\cbv$. 
A translation
 $(\cdot)_{\lambda}: \Lambda_{\oplus} \rightarrow \Lambda$ can be defined as follows, where $z$ is a fresh variable which is  used by no term:
	$$
	\begin{array}{lcl|lcl}
	(x)_{\lambda}&=&x & 	\tr{MN}&=&\tr{M}\tr{N}\\
	\tr{M \oplus N}&= &z\tr{M}\tr{N}&
	\tr{\lambda x. M}&=& \lambda x. \tr{M}\\
	\end{array}
	$$
The translation is injective (if $\tr M=\tr N$ then $M=N$) and preserves values.
\begin{prop}[Simulation]\label{prop:translation} The translation is sound and complete.
Let $M,N \in \Lambda_{\oplus}$. 
	\begin{enumerate}
		\item $M \redbv N$ implies $\tr M \redbv \tr N$;
		\item $\tr M \redbv Q $ implies there is a (unique) $N$, with $Q=\tr N$ and $M \redbv N$. 	
	\end{enumerate}
\end{prop}

%

 \subsection{Discussion (Surface Contexts)}\label{sec:surface} 
  The notion of surface context which we defined  is familiar in the setting of $\lambda$-calculus: it corresponds to  \emph{weak evaluation}, which  we  discussed in \ref{mot:result}. 
  In  $\PLambda^\cbv$, the $\redbv$-reduction is \emph{unrestricted}.
 Closing the $\oplus$-rules under surface  context $\ss$ expresses the fact that the   $\oplus$-redex is  not reduced  under  
 $\lambda$-abstraction, nor in the scope of another $\oplus$. The former  is fundamental to  confluence: it means that a  function which samples from a distribution  can be duplicated, but we cannot pre-evaluate the sampling. The latter is a technical simplification, which we adopt to  avoid unessential burdens with associativity. To require no reduction in the  scope of $\oplus$ is very  similar to allow no reduction in the branches of an if-then-else.

%
%

\section{Confluence and Standardization }\label{sec:conf_and_st}
\subsection{Confluence}\label{sec:confluence}
We prove that $\PLambda^\cbv$ is confluent.
We modularize the proof   using the Hindley-Rosen lemma. The notions of commutation and $\diamond$-commutation which we use are reviewed in Sec.~\ref{sec:rel}.
\begin{lemma*}[Hindley-Rosen]
	Let $\red_{1}$ and $\red_{2}$  be binary relations  on the same set $\mathcal R$. Their union $\red_{1}\cup\red_{2}$ is confluent if both $\red_{1}$ and $\red_{2}$ are confluent, and 
		 $\red_{1}$ and $\red_{2}$ commute.
\end{lemma*}

The following criterion  allows us to \emph{work pointwise}  in proving commutation and confluence of  binary relations on \emph{multidistributions}, namely $\Redbv$ and $\Redo$.
\begin{lemma}[Pointwise Criterion]\label{lem:pointwise}Let $\red_o,\red_b \subseteq \PLambda\times \MDST{\PLambda}$ and $\Red_o,\Red_b$ their lifting (as defined in \ref{def:lift}).
	Property (*) below implies that  $\Red_o,\Red_b$ $\diamond$-commute.
\begin{center}
		(*)	 If $M\red_b \n$ and $M\red_o \s, $ then  $\exists\r$ s.t. $\n\Red_o \r$ and $\s\Red_b \r$.
\end{center}
\end{lemma}

\begin{proof}We prove that (**) $\m\Red_b \n$ and $\m\Red_o \s$ imply  exists $\r$ s.t.  $\n\Red_o \r$ and $\s\Red_b \r$. 
Let $\m=\multiset{p_{i}M_{i}\mid \iI}$. By definition of lifting, for each $M_i$, we have $\mset{M_i}\Red_b\n_i$ and $\mset{M_i}\Red_o\s_i$, 
	with $\n=\sum p_i\cdot \n_i$ and $\s =\sum p_i\cdot \s_i$. 
	It is easily  checked,
	 \SLV{ using either  (*) or reflexivity,}{} that for each $M_i$, it exists $\r_i$ s.t. $\n_i\Red_o \r_i$ and $\s_i\Red_b \r_i$.
\SLV{}{If either $\mset{M_i}\Red_b \n_i$ or $\mset{M_i} \Red_o \s_i$ uses  reflexivity (rule $L1$), it is  immediate to obtain  $r_i$. Otherwise, $\r_i$ is given by 
	property (*).}
	Hence $\r = \sum_i p_i\cdot  \r_i$ satisfies (**). 
\end{proof}

We  derive confluence of $\Redbv$ from the same property in the CbV $\lambda$-calculus \cite{PlotkinCbV,RonchiPaolini}, using the simulation of Prop.~\ref{prop:translation}.

\begin{lemma}\label{lem:betaCR}
The reduction $\Redbv$  is confluent.
\end{lemma}
\SLV{}{
\begin{proof}
Assume $\m\Redbv^* \n$ and $\m\Redbv^* \s$.
	We first observe that if $\m=\mset{p_iM_i\mid \iI}$, then $\n$ and $\s$ are respectively of the shape $\mset{p_iN_i\mid \iI} $, $\mset{p_iS_i\mid \iI} $, with $M_i\redbv^*N_i$ and $M_i\redbv^*S_i$. By Prop. \ref{prop:translation}, we can project such reduction sequences on $\Lambda^\cbv$, obtaining that for each $i\in I$, $(M_i)_{\lambda}\redbv^* (N_i)_{\lambda}$ and $(M_i)_{\lambda}\redbv^* (S_i)_{\lambda}$.
Since $\redbv$ in CbV $\lambda$-calculus is confluent, there are $R_{i}\in \Lambda$ such that  $(N_i)_{\lambda} \redbv^{*} R_{i}$ and $(S_i)_{\lambda} \redbv^{*}R_{i}$. By Prop. \ref{prop:translation}.2, for each $i\in I$ there is a unique $T_{i} \in \Lambda_{\oplus}$ such that $(T_{i})_{\lambda}=R_{i}$, and the proof is given.	
\end{proof}
}

\SLV{We  prove  that   $\Redo$ is confluent (indeed, diamond).}{We  prove  that the reduction $\Redo$ is diamond, i.e., the reduction diagram closes in one step.}
\begin{lemma}\label{lem:oplusCR}
	The  reduction $\Redo$ is diamond.
\end{lemma}
\begin{proof}
We prove  that if $M\redo \n$ and  $M\redo \s$ , then  $\exists \r$ such that $\n\Redo \r$ and 
$\s\Redo \r$. The claim then  follows by  Lemma \ref{lem:pointwise}, by taking $\red_o~=~\red_b~=~\redo$.
%
Let $M=\ss(P\oplus Q) =\ss'(P'\oplus Q')$, $\n=\multiset{\frac{1}{2}\ss(P),\frac{1}{2}\ss(Q)} $ and
$\s=\multiset{\frac{1}{2}\ss'(P'),\frac{1}{2}\ss'(Q')}$. Because of  definition of surface context, the two $\oplus$-redexes do not overlap:  $P'\oplus Q'$ is  a subterm of $\ss$ and $P\oplus Q$ is  a subterm of $\ss'$. Hence we can reduce those redexes  in $\ss$ and $\ss'$,  to  obtain $\r$.
\end{proof}

We prove  commutation of  $\Redo$ and $\Redbv$ by proving a stronger property:  they  $\diamond$-commute. 
\begin{lemma}\label{lem:comm}The reductions 
 $\Redbv$ and $\Redo$ $\diamond$-commute.
\end{lemma}

\begin{proof}
By using  Lemma~\ref{lem:pointwise},  we only need to prove that 
if $M\redbv \n$ and $M\redo \s$, then $\exists \r$ such that  $\n\Redo \r$ and $\s\Redbv \r$.
The proof is by induction on $M$. 
\SLV{Observe that  $M$ cannot be  $\lam x.P$ or   $(\lam x.P') V$  because neither  can contain a  $\oplus$-redex.\\	
$\bullet$ Case $M= P\oplus Q$ is easy;   $M$ is the only possible  $\oplus$-redex.\\
$\bullet$	Case $M=PQ$. 	
 If the  $\beta_{v}$-redex is inside $P$, and the $\oplus$-redex inside $Q$, then
$PQ\redbv \multiset{P'Q}$ and  $PQ\redo \multiset{\two PQ',\two PQ''}$, with  $P\redbv P'$  and  $Q\redo \multiset{\two Q', \two Q''}$. Clearly  
$\r=\multiset{\two P'Q', \two P'Q''}$ satisfies the claim. 
 The dual case  is similar. 
 
 If  both redexes are inside  $Q$ (or $P$), we use induction. 
We write   $M$ as $\ss(Q)$; assume 
 $Q\redbv \mset{N}$, $Q\redo \mset{\two Q',	\two Q''}$, therefore   $\ss(Q)\redbv  \mset{\ss(N)}=\n$ and $\ss(Q)\redo \mset{\two \ss(Q'), \two \ss(Q'')}=\s$.  We 
 use the i.h. on $Q$ to obtain  $\r'=\mset{\two R', \two R''}$ such that 
 $\mset{N}\Redo\mset{\two R', \two R''}$, $\mset{Q'}\Redbv \mset{R'} $, $\mset{Q''}\Redbv \mset{R''} $.
 We conclude that  for  $\r=\mset{\two \ss(R'), \two \ss(R'')}$, it holds that $\n\Redo \r$ and $\s\Redbv \r$.
}
{
Cases $M=x$ and $M=\lam x.P$ are not possible given the  hypothesis.	
	\begin{enumerate}
		\item\label{case:sum} 
		Case $M= P\oplus Q$. $M$ is the only possible  $\oplus$-redex. Assume the $\beta_{v}$-redex is inside $P$ (the other case is similar), and that  $P\oplus Q\redbv \mset{P'\oplus Q}$, $P\oplus Q\redo \multiset{\two P,\two Q}$. It is immediate that  $\r=\multiset{\two P',\two Q}$ satisfies the claim.
		
		\item Case $M=PQ$. 	$M$ cannot have the form $(\lam x.P') V$ because neither $P$ nor $Q$ could contain a  $\oplus$-redex.
		\begin{enumerate}
			
			\item\label{case:disjoints} Assume that  the $\beta_{v}$-redex is inside $P$, and the $\oplus$-redex inside $Q$.
			We have   $PQ\redbv \multiset{P'Q}$ (with $P\redbv P'$), $PQ\redo \multiset{\two PQ',\two PQ''}$ (with $Q\redo \multiset{\two Q', \two Q''}$). It is immediate that 
			$\r=\multiset{\two P'Q', \two P'Q''}$ satisfies the claim. The symmetric case is similar.
			
\item\label{case:ind}
			Assume that both redexes are inside  $Q$.  Let us write   $M$ as $\ss(Q)$. Assume 
			$Q\redbv \mset{N}$, $Q\redo \mset{\two Q',	\two Q''}$, therefore   $\ss(Q)\redbv  \mset{\ss(N)}=\n$ and $\ss(Q)\redo \mset{\two \ss(Q'), \two \ss(Q'')}=\s$.  We 
			use the inductive hypothesis on $Q$ to obtain  $\r'=\mset{\two R', \two R''}$ such that 
			$\mset{N}\Redo\mset{\two R', \two R''}$, $\mset{Q'}\Redbv \mset{R'} $, $\mset{Q''}\Redbv \mset{R''} $.
			We conclude that  for  $\r=\mset{\two \ss(R'), \two \ss(R'')}$, it holds that $\n\Redo \r$ and $\s\Redbv \r$.

\end{enumerate}
	\end{enumerate}
}
\end{proof}

\begin{thm}
	The reduction $\Red$ is confluent.
\end{thm}
\begin{proof}
	By  Hindley-Rosen, from Lemmas \ref{lem:comm}, \ref{lem:betaCR}, and \ref{lem:oplusCR}.
\end{proof}

\SLV{Let us call $\n$  an $\Nnf$-multidistribution if $\n\in \MDST{\Nnf}$.}
{Let us call $\n$ an $\Nnf$-multidistribution if $\n\in \MDST{\Nnf}$
\ie\ $\n=\mset{p_iM_i}$  and all $M_i$ are $\red$-normal forms.}
The following fact is an immediate consequence of  confluence: 
\SLV{	\begin{center}
		\emph{"If   $\m$ reduces to a  $\Nnf$-multidistribution, it is unique."}
\end{center} }
{
\begin{fact*}\label{cor:unicity}
	The   $\Nnf$-multidistribution   to which  $\m$ reduces, if any, is unique. 
\end{fact*}
}

\subsubsection{Discussion}\label{discussion:conf}
While immediate, the  above  fact  is hardly useful, for two reasons. 
First, we know that probabilistic termination is  not necessarily reached in a finite number of steps; the relevant notion is not  that $\m\Red^*\n$ $\in \MDST{\Nnf}$,  but rather  that of a distribution which is defined as limit by the sequence $\seq \m$.
Secondly,  in  Plotkin's CbV calculus the result of computation is formalized by the notion of value, and considering normal forms as values is unsound (\cite{PlotkinCbV}, page 135).
In Section~\ref{sec:uniquelim} we introduce a suitable notion of limit distribution,  and  study   the implications of confluence on it.

\condinc{}{We will see, in Section \ref{sec:uniquelim}, that  confluence is a first step in proving that, in $\Lambda_{\oplus}$, the result of a \blue{probabilistically terminating} computation is still unique.   }

\condinc{}{
\subsection{Diamond Properties}
It is easy to revisit the proof in Sec.~\ref{sec:proof_conf} to establish the following

\begin{prop}[Diamonds] (1.) The  reduction $\sRed$ is diamond; (2.) the reduction $\sfull$ is diamond.
\end{prop}

\begin{proof}(1.) $\sRedbv$ 
	is diamond, because weak  $\redbv$ evaluation satisfies the following property: $b \leftarrow a \rightarrow c$ implies $ b=c$ or $\exists d$ s.t. $b \rightarrow d \leftarrow c  $.
	Together with Lemma.~\ref{lem:oplusCR} and Lemma~\ref{lem:comm}, this  immediately implies that $\sRed$ is diamond.
	
	(2.) is a consequence of the following facts.
	\begin{itemize}
		\item 
		$\sfullbv$ 
		is diamond, for the same reason  $\sRedbv$ is.
		\item If it is immediate to revisit the proof of Lemma~\ref{lem:oplusCR} to check that $\fullo$ is diamond.
		\item Similarly, if we go thorugh the proof  of Lemma~\ref{lem:pointwise}, it is immediate to check that 
		$\sfullbv$ and $\fullo$ 
		$\diamond$-commute. 
	\end{itemize}
\end{proof}
}


\condinc{}{
\CF{
Annoto  qui altri 2  punti - non necessariamente da inserire, piu' per noi \\
\begin{enumerate}
	\item 
	Confluence implies that
	the \emph{result} of the calculus is unique and independent from the evaluation order. 
	If for  result we mean "$M$ can reach a value", confluence implies that if $M\red^*V\in \Val$  and    $M\red^*P$ then $P$ also can reach a value.
	Nel caso del CbV, il risultato non e' tanto lo specifico valore, ma il fatto che un valore puo' essere raggiunto (eval e' difinito).
	
	\item Forse ci conviene restare con la formulazione neutra ma corretta "confluence implies that   for each term,  the result of a terminating computation is unique." A Limit distribution sembra un buon analogo di terminating computation.
\end{enumerate}	
}
}

\subsection{A Standardization Property}\label{sec:finitary_stCBV}

In this section, we first introduce  surface and left reduction as strategies for $\Red$. In the setting of the CbV $\lambda$-calculus,  the former corresponds to weak reduction,
the latter  to  the standard  strategy  originally defined in \cite{PlotkinCbV}.
We then establish  a   standardization result,  namely that every \emph{finite} $\Red$-sequence can be partially ordered as a sequence in which all  surface reductions are performed first. A counterexample shows that in $\PLambda^\cbv$, a  standardization result using left reduction fails.

\subsubsection{Surface and Left Reduction}\label{sec:strategies}

%
%

We remind the reader that in the $\lambda$-calculus, a \emph{deterministic strategy} defines   a function from terms to redexes,  associating to every term the next redex to be reduced. More  generally, we call reduction \emph{strategy}  for $\red$   a reduction \emph{relation}  $\red_a$ such that 
$\red_a \subseteq \red$. The notion of strategy can be easily formalized through the notion of context. With this in mind, let us consider  surface and left  contexts.

\begin{itemize}	
	\item  \emph{Surface contexts} $\ss$
	have been defined in Sec.\ref{def:language}.
	\item   \emph{Left contexts}  $\lc$ are defined by the following grammar:
	\[ \lc ::= \square \mid \lc M \mid V \lc \]
	Note that in particular a left contexts is a surface context.
	\item We call \emph{surface reduction}, denoted by $\sred$ (with lifting $\sRed$) and \emph{left reduction}, denoted by $\lred$ (with lifting $\lRed$), the closure of 
	the reduction rules  in Fig.~\ref{fig:rules} under surface contexts and left contexts, respectively.  It is clear that   $\sred ~= ~\sredbv \cup \redo$. Observe that $\lred \subsetneq\sred$.
	\item  A reduction step  $M\red \m$ is \emph{deep}, written $M\dred \m$, if it is \emph{not}  a surface step. A reduction step is  \emph{internal} (written $M\ired \m$) if it is \emph{not} a left step. Observe that $\dred \subset \ired$.	
\end{itemize}

\begin{example}
	\SLV{
($\lred \subsetneq\sred$)
Let $M = x(II)(II)$, where $I=\lambda x.x$. Then $M \sred \multiset{x I(II)}$ and $M \sred \multiset{x (II)I}$; instead, 
$M \lred \multiset{x I(II)}$,   $M \not\lred \multiset{x (II)I}$.}
{
\begin{itemize}	
\item ($\lred \subsetneq\sred$)
	Let $M = x(II)(II)$, where $I=\lambda x.x$. Then $M \sred \multiset{x I(II)}$ and $M \sred \multiset{x (II)I}$; instead, 
	$M \lred \multiset{x I(II)}$,   $M \not\lred \multiset{x (II)I}$.
	\item ($\dred \subsetneq \ired$) Let $M=(\lambda x. II)(II)$. Then $M\ired (\lambda x. I)(II)$ and $M \ired (\lambda x. II)I$, while 
	$M\dred (\lambda x. I)(II)$ and $M \not\dred (\lambda x. II)I$
		\end{itemize}
	}
\end{example}
Intuitively, left reduction chooses the leftmost of the surface redexes. More precisely, this is the case for closed terms (for example, the term $(xx)(II)$ has a $\sred$-step, but no  $\lred$-step).

\emph{Surface  Normal Forms:}  
We denote by   $\Snf^\cbv$    the set of $\sred$-normal forms. We observe that all values are surface normal forms (but the converse does not hold):
$\Val \subsetneq \Snf^\cbv$\SLV{.}{ (and $\Nnf\subsetneq \Snf^\cbv$).} 
The situation is different if we restrict ourselves to close term, in fact the following result holds, which is easy to check.
	\begin{lemma}\label{lem:closed}If  $M$ is a \emph{closed} term, the following three are equivalent: 
		\SLV{		(i) $M$ is a $\sred$-normal form; (ii)	$M$ is a $\lred$-normal form; (iii) $M$ is a value.}
		{\begin{enumerate}
				\item $M$ is a $\sred$-normal form;
				\item $M$ is a $\lred$-normal form;
				\item $M$ is a value.
		\end{enumerate}}
	\end{lemma}

%

\subsubsection{Finitary Surface Standardization}\label{sec:finitaryst}
The next theorem proves a standardization result,  in the sense that every \emph{finite} reduction sequence can be (partially) ordered in  a sequence of surface  steps followed by a sequence of deep steps.

\begin{thm}[Finitary Surface Standardization]\label{thm:surfacestandard} In $\PLambda^\cbv$, if $\m\Red^*\n$ then exists $\r$ such that $\m \sRed^* \r$ and $\r \dRed^* \n$.
\end{thm}
\begin{proof}
	We build on an  analogous  result for CbV $\lambda$-calculus,  which is folklore and is  proved explicitly  in 
	\SLV{\cite{long}}{Appendix~\ref{sec:finitary_stCBV}}. We then only need to check that deep steps commute with  $\oplus$-steps, which is straightforward technology \SLV{(\cite{long} gives the full proof)}{(the full proof is in  Appendix~\ref{sec:finitary_stCBV})}.
\end{proof}

\paragraph{Finitary Left Standardization does not hold}	
The following statement  \emph{is false} for $\PLambda^\cbv$.
\begin{center}
	\emph{``If $\m\Red^*\n$ then there exists $\r$ such that $\m \lRed^* \r$ and $\r \iRed^* \n$."}
\end{center}

\begin{example}[Counter-example]\label{counterex:standard}Let us consider the following sequence,  where $I=\lam x.x$ and
	 $M=(II)((\lam x. y\oplus z)I)$.
	$\mset{M}\iRed \mset{(II)(y\oplus z )}$$\Red_{\oplus} \mset{\two (II)y, \two (II)z} \Redbv \mset{\two Iy,\two (II)z}$. 
	If we  anticipate the reduction of $(II)$, we have $M\lredbv \mset{I ((\lam x.y\oplus z)I)}$, from where we cannot reach $\mset{\two Iy,\two (II)z}$. Observe that the sequence  is already  surface-standard!
\end{example}

\section{Asymptotic Evaluation}\label{sec:asymptotic}
The specificity of probabilistic computation is to be concerned with asymptotic behavior; the focus is not  what happens 
after a finite number $n$ of steps, but when $n$ tends to infinity.
In this section, we study  the  asymptotic behavior of $\Red$-sequences with respect to evaluation.  
The intuition is that  a reduction sequence defines a distribution on  the possible outcomes of the program.
We first  clarify what is the outcome of evaluating a probabilistic term, and then 
we formalize the idea of  result ``at the limit" with  the notion of \emph{limit distribution} (Def.~\ref{def:limits}). In Sec.~\ref{sec:uniquelim} we  investigate 
how   the asymptotic result of different sequences starting from the same $\m$ compare.\\

We recall that to each multidistribution $\m$ on $\PLambda$  is  associated a probability distribution $\mu\in \DST{\PLambda}$   (see Sec.\ref{sec:multi}). We   use  the following \textbf{letter convention}: given a multidistribution
$\m,\n,\r, ...$ we denote the associated distribution by the corresponding Greek letter $\mu,\nu,\rho, ...$
If $\seq{\m}$  is  a $\Red$-sequence, then  $\seq \mu$ 
is the sequence of associated distributions.

\subsection{Probabilistic Evaluation}
We start by studying  the property of being valuable (\ref{sec:valuable}) and by analyzing some examples (\ref{sec:resultCbV}). This  motivates the more general approach we introduce in \ref{sec:limit}.

\subsubsection{To be valuable}	\label{sec:valuable}
In the CbV $\lambda$-calculus, the key property of a term $M$ is \emph{to be valuable, i.e., $M$ can reduce to a value. 
	To be valuable is a \emph{yes/no} property,  whose  probabilistic  analogous  is \emph{the probability to reduce to a value}. }
If  $\m$ describes the result of a computation step, the probability that such a result is a value is simply 
$\mu (\Val):=\sum_{V\in \Val} \mu(V)$, \ie\ the probability of the event $\Val\subset \PLambda$.
Since the set of values is closed under reduction, the following property holds: 
\begin{fact}\label{rem:value}
	If $V\in \Val$ and  $V\red \m$, then $\m = \mdist{W}$, with $W\in \Val$, and $V \redbv \mdist{W}$.
\end{fact}

%

Let $\seq{\m}$  be a $\Red$-sequence, and $\seq \mu$ the sequence of associated distributions.
The sequence  of reals $\langle{\mu_n(\Val)}\rangle_{n\in \Nat}$ is nondecreasing  and bounded, 
because of Fact~\ref{rem:value}. 
Therefore \emph{the limit exists}, and is the supremum:  $\lim_{n\to \infty} {\mu_n(\Val)}=\sup_n\{\mu_n(\Val)\}.$
This fact allows us  the following definition.

\begin{itemize}
	\item The sequence	$\seq{\m}$ \emph{evaluates with probability $p$}\\  if $p = \sup_n {\mu_n(\Val)} $,  written
	$\seq{\m} \tolim p$.
	\item $\m$ is \emph{ $p$-valuable} if $p$ is the greatest probability to
	which a sequence from $\m$ can evaluate.
	
\end{itemize}

\begin{example}\label{ex:evaluation}
	Let $\true=\lambda xy.x$ and $\false=\lambda xy.y$. 
	\begin{enumerate}
		\item Consider the term $PP$ where $P=(\lam x.(xx\oplus \true))$. Then $PP \red \multiset{(PP)\oplus \true}\Red \multiset{\frac{1}{2}PP, \frac{1}{2}\true} 
		\Red^{2n}\multiset{\frac{1}{2^{n}}PP, \frac{1}{2}\true, \dots, \frac{1}{2^{n}} \true}$ . 
				Since $\lim_{n\to \infty} { \sum_{1}^{n} \frac{1}{2^{n}}}=1$, $PP$ is $1$-valuable.

		\item Consider the term $QQ$, where  $Q=\lam x.(xx\oplus (\true\oplus \false))$. 
	\SLV{}{	Then $QQ \redbv \multiset{(QQ)\oplus (\true\oplus \false)  } \Red^*
		\multiset{\frac{1}{2}QQ, \frac{1}{4}\true, \frac{1}{4}\false}\Red^*\dots$ }
		It is immediate that $QQ$ is $1$-valuable.
		
		\item Let $\Delta =\lam x.xx$, so that $\Delta\Delta$ is a divergent term, and let
	 $N=\lam x.(xx)\oplus (\true\oplus (\Delta\Delta))$.
		\SLV{}
		{Then
		$NN \redbv \multiset{(NN)\oplus (\true\oplus (\Delta\Delta))  } \Red^* \multiset{\frac{1}{2}NN, \frac{1}{4}\true, \frac{1}{4}(\Delta\Delta)}\Red^*\dots$ }
		$NN$ is $\two$-valuable.
	\end{enumerate}
\end{example}


\subsubsection{Result of a CbV computation}\label{sec:resultCbV} The notion of  \emph{being  $p$-valuable} allows for a simple definition, but it is too coarse. 
Consider  Example \ref{ex:evaluation}; both   1) and 2) give  examples of  $1$-valuable term. However, in 1) the probability is concentrated in the value $\true$, while in 2)  $\true$ and $\false$ have equal probability $\two$. Observe that $\true$ and $\false$ are different normal forms, and are not  $\beta_{v}$-equivalent. To discriminate between $\true$ and $\false$, we need  a finer notion of evaluation. 
Since the calculus is CbV, the result ``at the limit" is intuitively a distribution on  the possible values that the term can reach.
Some care is needed though, as the following example shows. 

\begin{example}\label{ex:nonnf} Consider Plotkin's CbV   $\lambda$-calculus. 
	Let 
	$\omega_{3}=\lam x.xxx$; the term $M={(\lam x.x) \lam x.\omega_{3}\omega_{3}}$  has the following $\redbv$-reduction:  $M={(\lam x.x)( \lam x.\omega_{3}\omega_{3})} \redbv M_1={ \lam x.\omega_{3}\omega_{3}}\redbv M_2= { \lam x.\omega_{3}\omega_{3} \omega_{3}} \redbv\cdots$. We obtain a reduction sequence   where  $\forall n\geq1$,  $M_n= {\lam x. \omega_{3}\underbrace{\omega_{3}...\omega_{3}}_{n}}$.
	Each $M_i$ is a value, but there is not a "final" one in which the reduction ends. 
	Transposing this to  $\PLambda^\cbv$, let $\m_0=\mset{M}$, $\m_i=\mset{M_i}$. The $\Red$-sequence $\seq \m$ is  $1$-valuable, but the distribution on values is different at every step.	In other words, $\forall V\in \Val$,  the  sequence $\langle{\mu_n({V})}\rangle$ has limit $0$.
	Observe 
	that however all the values $M_i$  are $\beta_{v}$-equivalent.
	
\end{example}

\subsubsection{Observations and Limit Distribution}\label{sec:limit}
Example~\ref{ex:nonnf} motivates the approach that we  develop now: 
the result of probabilistic evaluation  is not a distribution on values, but\emph{ a distribution on some events of interest}. In the  case of $\PLambda^\cbv$, the most informative events are equivalence classes of values. 

We first introduce the notion of observation, and then that of limit distribution.

\begin{Def}\label{def:obs} 
A set of \emph{observations} for  $(\PLambda,\Red)$ is a set  $\Obs\subseteq\mathcal P(\PLambda)$ such that  $\forall \bU,\bZ\in \Obs $, if $ \bU \not=\bZ $ then  $\bU\cap\bZ=\emptyset$, and 	if $\m\Red \m'$  then $\mu(\bU) \leq \mu'(\bU)$.
\end{Def}
\condinc{}{A set of \emph{observations} for  $(\PLambda,\Red)$ is a set  $\Obs\subseteq\mathcal P(\PLambda)$ such that  $\forall \bU,\bZ\in \Obs$,   $\bU\cap\bZ=\emptyset$ and 	if $M\in \bU$ and $M\red \mset{p_iM_i\st \iI}$, then $M_i\in \bU$, $(\forall \iI)$.
}
Note that, given  $\mu\in \DST{\PLambda}$,  $\bU\in \Obs$  has probability  $\mu(\bU)$ (similarly to the event "the result is Odd" in Example~\ref{ex:odd}).

{It follows immediately from the definition that,  given a  sequence $\seq{\m}$, then for each $\bU\in \Obs$ the sequence $\langle{\mu_n(\bU)}\rangle_{n\in \Nat}$ is nondecreasing and bounded, and therefore has a limit,  the $\sup$. 
Moreover, monotony implies the following \begin{equation}\label{eq:MCT}
	{ \sup_{n} \{\sum_{\bU\in \Obs} \mu_n(\bU)\} ~=~  \sum_{\bU\in \Obs} \sup _{n}  \{\mu_n(\bU)\}.}
\end{equation}
which guarantees that the distribution $\brho$ in Def.~\ref{def:limits}   is  well defined, because $ \sup_n \norm {\mu_n}\leq 1$ and  (1) gives $ \sup_n \norm {\mu_n}=\norm \brho$.
}
%
\begin{Def}\label{def:limits} Let $\Obs$ be a set of observations. 
	 The sequence 	$\seq{\m}$  defines a distribution $\brho\in \DST{\Obs}$, where   $\forall \bU\in \Obs$,
\begin{center}
		  $ \brho(\bU):= \sup_n \{\mu_{n}(\bU)\}. $
\end{center}
\begin{itemize}
\item We call such a $\brho$ the\emph{ limit distribution} of $\seq{\m}$. \textbf{Letter convention}: greek bold letters   denote  limit distributions.  

	\item The sequence 	$\seq{\m}$   \emph{converges to (or evaluates to)} the limit distribution 
		$\brho$,   written 
\begin{center}
	$\seq{\m} \down_{{}_\Obs}\brho$.
\end{center}

{\item If $\m$ has a sequence  which converges to $\brho$, we  write 
\begin{center}
	$\m\tolimx{\Obs} \brho$.
\end{center}}

	\item Given $\m$, we denote by  $\Lim_{{}_\Obs}(\m)$ the set $\{\brho\st \m\tolimx{\Obs} \brho\}$  of all  limit distributions of $\m$.
	If $\Lim_{{}_\Obs}(\m)$ has a greatest  element, we indicate it by  $\den{\m}_{{}_\Obs}$. 
\end{itemize}
If $\Obs$ is clear from the context,  we omit the index which specifies it, and simply write $\seq{\m} \down\brho$, $\m\tolim \brho$,
$\Lim(\m)$.
\end{Def}{}

The notion of limit distribution formalizes what is the \emph{result of evaluating} a probabilistic term, once we choose the set $\Obs$ of  observations which interest us.  In  \ref{sec:uniquelim} we  prove  that confluence implies that   $\Lim (\m)$ has  a unique maximal element.

%
%
\paragraph{Sets of Observations for $\PLambda^\cbv$} 
Let us consider two partitions of the set  $\Val \subset \PLambda$,  the trivial one  $\{\Val\}$,  and   the set $\Val_\sim$ of 
values up to the equivalence $\equbetav$,  \ie\  the collection of all events $\{W\in \Val \st W\equbetav  V\}$.  For the set $\Nnf$ of  $\red$-normal forms  (see \ref{def:nf}), interesting partitions   are $\{\Nnf\}$ and 
the set of singletons $\Nnf_\sing:=\{\{M\}, M\in \Nnf \}$.
\begin{prop}\label{prop:observations}
	$\{\Val\}$, 	$\Val_\sim$, $\{\Nnf\}$ and $\Nnf_\sing$ are each a set of observations for $(\PLambda, \Red)$.
\end{prop}
\begin{proof}Clearly, any partition of  $\Nnf$ satisfies the conditions in Def.~\ref{def:obs}.
	For $\{\Val\}$ and $\Val_\sim$, the result  follows from Fact~\ref{rem:value}.
\end{proof}
Notice  that convergence w.r.t.   $\{\Val\}$ corresponds to the notion  of being $p$-valuable. Instead   $\{\Nnf\}$ and   $\Nnf_\sing$  correspond to \SLV{}{normalization and reaching a specific normal form, respectively; however these are}   events which are \emph{not  significant}  in a CbV perspective, as we already discussed  in \ref{discussion:conf}. For this reason, in Sec.~\ref{sec:left_evaluation}
we will focus  on the study of $\Obs:=\Val_\sim$ (Sec.~\ref{sec:left_evaluation}). 

\begin{example}\begin{itemize}
		\item 
  Let  $\Obs$ be either $\Val_\sim$ or  $\Nnf_\sing$.
		\begin{enumerate}
			\item Let $\seq{\m}$ be the sequence in Example \ref{ex:evaluation}.1, starting from $\multiset{PP}$.
			Then $\seq{\m}\down_{{}_\Obs}\{\pmb{\true}^{1}\}$.
			\item Let $\seq{\m}$ be the computation in Example \ref{ex:evaluation}.2, starting from $\multiset{QQ}$.
			Then $\seq{\m}\down_{{}_\Obs}\{\pmb{\true}^{\two}, \pmb{\false}^{\two}\}$.
			\item Let $\seq{\m}$ be the computation in Example \ref{ex:evaluation}.3, starting from $\multiset{NN}$.
			Then $\seq{\m}\down_{{}_\Obs}\{\pmb{\true}^{\two}\}$.
		\end{enumerate}
		\item  Let $\seq{\m}$ be the reduction sequence in Example \ref{ex:nonnf}, starting with $\multiset{(\lam x.x) \lam x.\omega_{3}\omega_{3}}$. By taking as set of observations  $\Val_\sim$,  the sequence  converges to $\{\pmb{\lam x.\omega_{3}\omega_{3}}^{1}\}$.
		 	\end{itemize}
\end{example}


\paragraph{Discussion}\label{rem:result} 
Each observation expresses  a result of interest for the evaluation of the term $M$. To better understand this, 
let us examine what become our notions of observation in the case of usual (non-probabilistic) CbV $\lambda$-calculus.
Let $M\red^*N\in \bU$ and $\bU\in \Obs$; if $\bU\in \{\Val\}$ then $M$ is valuable, if $\bU\in \Val_\sim$, then $M$ reduces to the value $N$ up to $\beta_{v}$-equivalence, if
 $\bU\in\{\Nnf\}$, then $M$ normalizes, finally  $\bU=\{N\}\in \Nnf_\sing$ means that $M$ has normal form  $N$.
 We say that  $\bU\in \Obs$ is a result of evaluating $M$,  if $M\red^*N\in \bU$.
Clearly,  fixed $\Obs$, confluence implies  that the result of evaluating $M$, if any, is unique.

\condinc{}{Observe that a common property to each such set, is that  $\forall \bU,\bZ\in \Obs$,  $\bU\cap\bZ=\emptyset$ and 	if $M\in \bU$ and $M\redbv M'$, then $M'\in \bU$.
%
}

\SLV{}{	
	\paragraph{Sets of observations for Surface Reduction}	It is interesting to examine   the \emph{set of observations for  surface reduction} $\sRed$.
	When considering $\sred$, values are $\sred$-normal forms (the converse does not hold!). Therefore  $\set{\{V\} \st V\in \Val }$ (where $\{V\}$ is a singleton) is a set of observations for $(\PLambda,\sRed)$.  In other words, when restricting oneself to surface reduction,  the result of a probabilistic computation (\ie\ the limit distribution) is a distribution on the possible values of the term.
	Observe that all set of observations for $\Red$ (Prop.~\ref{prop:observations}) are also set of observations for $\sRed$.
}

\subsection{Uniqueness and Adequacy of the Evaluation}\label{sec:uniquelim}
In this section, we adapt  similar results from \cite{pars}, to which we refer for 
 details. We  assume  a set $\Obs$ to be fixed, hence we omit  the index. For concreteness,  think of  $\Val_\sim$, but the results only depend on the properties  
in Def.~\ref{def:obs}, and on confluence.

How do different reduction sequences from the same initial $\m$ compare? More precisely, assume $\m \tolim \brho$ and $\m \tolim \bmu$, how do $\brho $ and $\bmu$ compare? 
Intuitively,  the limit distributions of $\m$ (which are the result of a \emph{probabilistically terminating} sequence)  play the role of normal forms in finitary termination.  As confluence implies uniqueness of normal forms, a  similar property holds when considering \emph{probabilistic termination} and limits, 
in the sense that each $\m$ has a \emph {unique {maximal} limit distribution} (Thm.~\ref{thm:unique}). While the property is similar, the  proof   is not as immediate as in the finitary  case.
The key  result   is  Lemma~\ref{lem:NFP} which implies  both that $\Lim(\m) $ has a greatest element (Thm.~\ref{thm:unique}), and adequacy of the evaluation (Thm.~\ref{thm:adequacy}). 

Recall that  the order  $ \leq$  on distributions is defined  pointwise (Sec.~\ref{sec:proba}).

\begin{lemma}[Main Lemma]\label{lem:NFP}$\PLambda^\cbv$ has the following property:
	$\forall \m,\s$,	if   $\bmu \in \Lim(\m)$, and   $\m \Red^*\s$,  then $\s\tolim \bsigma$ with  $\bmu \leq \bsigma$.	
	 Moreover, if $\bmu$ is maximal in $\Lim(\m)$ then $\bsigma =\bmu$.
\end{lemma}

\begin{proof} Let  $\bmu\in \Lim(\m)$, and $\seq \m$ be a sequence from $\m=\m_0$ which converges to $\bmu$.   Assume  $\m\Red^*\s$. As illustrated in Fig.~\ref{fig:nfp}, 
	 from $\s$ we build a sequence    $\s=\s_{\m_0}\Red^* \s_{\m_1} \Red^*\s_{\m_2}\dots$, 
	where each  segment $\s_{\m_{i}}\Red^*\s_{\m_{i+1}}$  ($i \geq 0$) is    given by confluence from  $\m_i\Red^*\s_{\m_{i}}$ and $\m_i\Red \m_{i+1}$. 	Let  $\seq \s$ be the  concatenation of  all such segments 
	and let  $\bsigma $ be its limit distribution. Clearly, $\bsigma \in \Lim(\m)$. Since by  construction $\m_i\Red^* \s_{\m_i}$,
	then 	   for each $\bV\in \Obs$, $ {\mu_i(\bV)}    \leq \bsigma(\bV)$ (because $\mu_i(\bV)\leq  {\sigma_{\m_i}(\bV)} $ by  definition of observation). 
	Therefore 
	$\sup_n \{\mu_n (\bV)\}=\bmu(\bV)\leq \bsigma(\bV) $. If $\bmu$ is maximal, then  $\bsigma= \bmu$.
\end{proof}

\begin{thm}[Greatest Limit Distribution]\label{thm:unique}
	$\Lim(\m)$ has a greatest element, which we indicate by  $\den \m$.
\end{thm}

\begin{proof}The proof of both existence and uniqueness of  maximal elements relies on Lemma~\ref{lem:NFP}.
	 Let us explicitly show  uniqueness. 
Let $\bmu\in \Lim(\m)$ be  maximal. Given any $\brho\in  \Lim(\m)$, we prove that $\brho\leq \bmu$. Let $\seq \r$ be a sequence from $\m$ such that $\seq \r\down\brho$. By Lemma \ref{lem:NFP},   $\forall \r_n$  there is a  $\Red$-sequence  from  $\r_n$ which has limit $\bmu$. Therefore    $\forall \bV\in \Val$, $\forall n$, $\rho_n(\bV)\leq \bmu (\bV)$, hence  
	$\brho(\bV)\leq \bmu(\bV) $. If $\brho$ is maximal, $\brho=\bmu$.
\end{proof}

\SLV{
	\begin{thm}[Adequacy]\label{thm:adequacy} If $\m \Red^* \s$, then	$\den\m = \den\s$.
	\end{thm}
}
{
	\begin{thm}[Adequacy of evaluation]\label{thm:adequacy}\quad\\ If $\m \Red^* \s$, then	$\den\m = \den\s$.
	\end{thm}
}

\begin{proof}  Observe first  that  $\den \s \in \Lim(\m)$,   hence   $\den \s \leq \den \m$.
	\SLV{}{ Indeed,  if $\seq \s \down\den \s$,
		by  concatenanting $\m \Red^* \s$ with $\seq \s$, 
		we have $\m \tolim \den \s$. }   By Lemma \ref{lem:NFP}, it  holds that $\den \m \in \Lim(\s)$, hence $\den \m \leq \den \s$. \SLV{}{Therefore $\den\m = \den\s$. }
\end{proof}

\SLV
{
	\vspace*{-8pt}
	\begin{figure}[h]\centering
		\begin{minipage}[l]{0.23\textwidth}		
			\fbox{		\includegraphics[width=\textwidth]{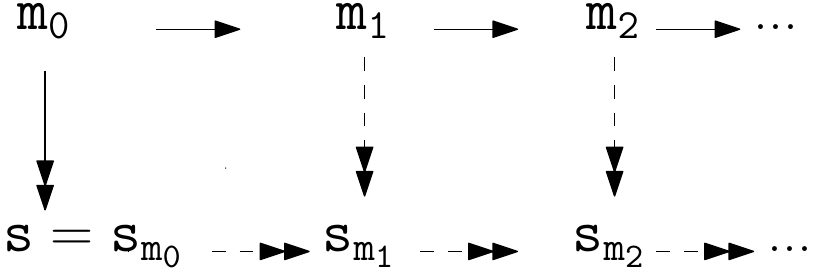}}
			\caption{Proof of Main Lemma}\label{fig:nfp}
		\end{minipage}
		\hfill
		\begin{minipage}[r]{0.23\textwidth}
			\fbox{		\includegraphics[width=1\textwidth]{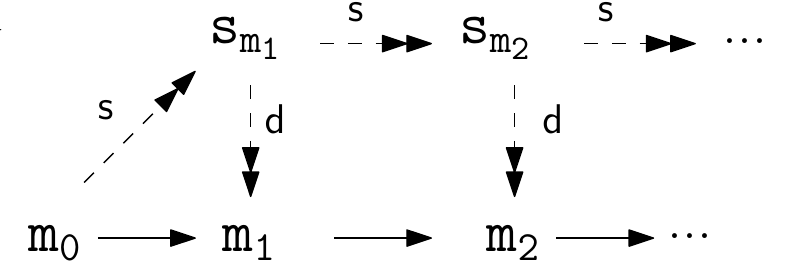}}
			\caption{Surface evaluation}\label{fig:surfEV}	
		\end{minipage}
	\end{figure}
	\vspace*{-8pt}
}
{
	\begin{figure}[h]\centering
		\begin{minipage}[c]{0.35\textwidth}		
			\fbox{		\includegraphics[width=\textwidth]{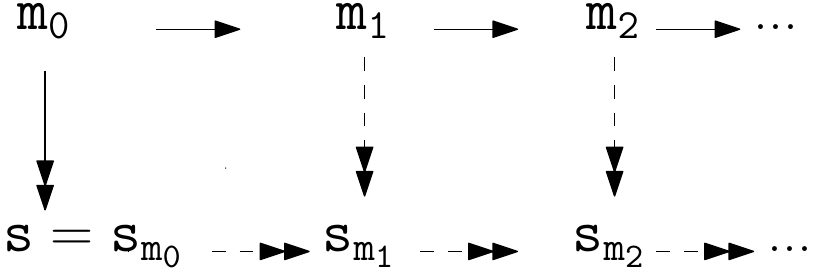}}
			\caption{Proof of Main Lemma}\label{fig:nfp}
		\end{minipage}
		
		\vspace*{6pt}
		
		\begin{minipage}[c]{0.35\textwidth}
			\fbox{		\includegraphics[width=1\textwidth]{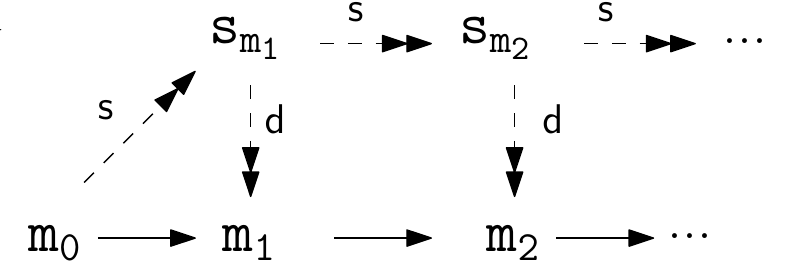}}
			\caption{Surface evaluation}\label{fig:surfEV}	
		\end{minipage}
	\end{figure}
	\vspace*{-4pt}
}

\section{Asymptotic Standardization}\label{sec:left_evaluation}

\condinc{}{
\CF{OSSERVAZIONE: mi hai convinta che, come dicevi,   non abbiamo bisogno di restrizione ai chiusi...  la restrizione e' "embedded" nella definizione di convergenza (che filtra i valori). Semplicemente:
	\begin{center}
		$x(II)\down_{{}_{\Val_{\sim}}} {\bm 0} $ (la distribuzione che ha valore 0 ovunque) iff $\L \down {\bm 0} $ (con le definizioni di
		Thm\ref{thm:eval_final})
	\end{center} Mi sembra che la situazione sia un po' diversa rispetto alla left di lambda CbV, la cui correttezza richiede  termini chiusi
(Su $M$, $\lred$ termina iff  $M$ is valuable) ?}
}

In this section, we focus on  $\Val_\sim$ as set of observations, which is the most natural choice in a CbV setting, in particular if we want to evaluate \emph{programs}, i.e., closed terms.

We proved, in Thm.~\ref{thm:unique}, that each $\m$  has a unique maximal limit distribution $\den \m$. Now we address the question: is there a reduction strategy which is guaranteed to converge to  $\den \m$?
We show that surface evaluation  provides such a strategy; indeed, \emph{any} limit distribution in $\Lim(\m)$ can be reached by surface evaluation (Thm.~\ref{thm:surfeval}). This result of \emph{asymptotic completeness}  is the main technical contribution of the section. \\ 

Following the notation introduced in \ref{sec:limit}, we  denote by $\bV$ the set $\{W\in \Val \st W\equbetav  V\}$. 
 We observe that:
\begin{fact}  
  Let  $M\dred \m$, then $\m$ has form $\mset{P}$ and $M =_{\beta_v} P$; 
		  $M$ is a value if and only if $P$ is a value.
\end{fact}
As a consequence of the previous fact, we have  
\begin{lemma}\label{lem:int}
	If $\m\dRed \s$ then $\mu(\Val)=\sigma (\Val)$, and $\mu(\bV)=\sigma(\bV)$, for each $\bV\in \Val_\sim$.
\end{lemma}

We write $\m \tolims \bmu$ (resp. $\m \toliml \bmu$) if there is a sequence $\seq \m$ such that all steps $\m_i\Red \m_{i+1}$ are surface (resp. left) reductions
and $ \seq \m \down\bmu$.
Remember that given $\m$, we write   $\den \m$ for  the unique maximal element of $\Lim(\m)$, and  $\m\tolim \bmu$ if there \emph{exists} a $\Red$-sequence from $\m$ which converges to $\bmu$.



We now prove asymptotic completeness for  surface evaluation.  
We  exploit  finitary  standardization  (Thm.~\ref{thm:surfacestandard}) and  \emph{extend it to the limit}. In the proof, it is essential  the fact that $\dRed$-steps preserve the distributions (Lemma~\ref{lem:int}).

\begin{thm}[Asymptotic Completeness of Surface Reduction]\label{thm:surfeval}	$\m\tolim\bmu$ if and only if   $\m \tolims \bmu$.
\end{thm}

\begin{proof} We prove  that $ \m \tolim \bmu$ implies $ \m \tolims \bmu$ (the other direction holds by definition). 
		Assume $\seq{\m} \down\bmu $, with $\m=\m_0$. As illustrated in Fig.~\ref{fig:surfEV},
	we  build  a sequence $\langle \s_{\m_n}\rangle $ such that  $\m_0=\s_{\m_0}$ and  $\forall i$ ($ \s_{\m_i}\sRed^* \s_{\m_{i+1}}$  and $\s_{\m_{i+1}}\dRed^* \m_{i+1}$).   If $i=0$, 
	 by Thm.~\ref{thm:surfacestandard} it exists  $\s_{\m_1}$  such that $\m_0=\s_{\m_0}\sRed^* \s_{\m_1}\dRed^* \m_{1}$.
	We then procede by induction: for  each $i>0$, we apply  Thm.~\ref{thm:surfacestandard} to the sequence  $\s_{\m_{i}}\dRed^* \m_{i}\Red \m_{i+1}$,  and obtain 
	the multidistribution $\s_{\m_{i+1}}$ such that  $ \s_{\m_i}\sRed^* \s_{\m_{i+1}}$  and $\s_{\m_{i+1}}\dRed^* \m_{i+1}$, as wanted. 
	The concatenation of  all segments  $\s_{\m_0}\sRed^* \s_{\m_1}, ..., \s_{\m_i} \sRed^*\s_{\m_{i+1}} , ...$ is a $\sRed$-sequence. Let $\bsigma$ be its limit.
	By Lemma \ref{lem:int} and the fact that  $\s_{\m_{i}}\dRed^* \m_{i}$, we have   $\sigma_{\m_{i}}(\bV) = \mu_i(\bV)$, for each $\bV\in \Val_\sim$. We conclude  $\bsigma=\bmu$ because $\forall i$:
	\begin{enumerate}
		\item $\sigma_{\m_{i}}(\bV) = \mu_i(\bV)\leq \bmu(\bV)$, therefore $\bsigma(\bV) \leq \bmu(\bV)$.
		\item $ \mu_i(\bV)= \sigma_{\m_{i}}(\bV) \leq \bsigma(\bV)$, therefore $\bmu(\bV) \leq \bsigma(\bV)$.
	\end{enumerate}
\end{proof}

\begin{remark}We observe that  completeness of  surface evaluation  (Thm.~\ref{thm:surfeval}) is specific to convergence w.r.t. $\Val_{\sim}$  and  $\{\Val\}$  (the most natural set of observations in CbV).
	\SLV{It does not necessarily hold}{Surface evaluation  is not  necessarily complete} if we evaluate  w.r.t. other sets of observations, such as normal forms, where deep  steps may be needed. Consider, for example,  the term 
	$\lambda z. II \dred \lambda z. I$.
	To define a complete strategy  w.r.t. $\Nnf_\sing$   demands to  refine the approach.
\end{remark}

\subsection{Surface and Left  Evaluation}\label{sec:fullcomplete}
We are now equipped to tackle  the goal of this section, namely the existence of  a strategy to find the greatest limit distribution of a program.

Since our aim is to   reach the greatest limit, it makes sense to reduce "whenever is possible", and use the full lifting $\full$ (Def.~\ref{def:fulllift}).
The reason is easy to see. Consider for example  $\m=\mset{\two \Delta\Delta, \two II}$, which has greatest limit  $\den \m =\{\mathbf{I}^\two\}$. We observe  that a $\Red$-sequence from $\m$ may very well keep reducing only the diverging term $\Delta\Delta$ and  never reach $\den\m$. The reduction $\full$, instead, forces the reduction of each term which is not in normal form for $\red$.
\begin{lemma}Let $\brho$ be maximal  among the limit distribution of all   $\full$-sequences from $\m$. Then $\brho =\den{\m}$. 
\end{lemma}
\begin{proof}
	Obviously, $\brho \in \Lim(\m)$. It is straightforward to check that if $\bmu$ is the limit of a $\Red$-sequence, then  there is a $\full$-sequence, whose limit is greater or equal to  $\bmu$.
\end{proof}
 We write $\sfull$ (resp. $\lfull$) for the full lifting of $\sred$  (resp. $\lred$). Observe that given $\m$, there is only one $\lfull$-sequence.
We use  the letters $\L=\seq l, ~\S=\seq s, ~\T=\seq t$  to indicate (infinite) reduction sequences.
 We say that $\m$ is \emph{closed} if it is a multidistribution  on closed terms\SLV{.}{ \ie\ 
$\m=\mset{p_iM_i\st \iI}$ with $M_i$ closed $\forall \iI$.}
\begin{prop}[Left Evaluation]\label{thm:lefteval} 	Let $\m$ be closed.
	\begin{enumerate}
		\item 	Let $\S, \T$ be  $\sfull$-sequences from $\m$; $\S\down \bmu$ if and only if $\T \down \bmu$.
		
		
		\item Let  $\S $ be  any $\sfull$-sequence from $\m$, and $\L$   the $\lfull$-sequences from $\m$. Then $\S\down \bmu$ if and only if $\L\down \bmu$.
	\end{enumerate}
\end{prop}
\begin{proof}
 \cite{pars}, Sec.6,  studies a CbV probabilistic $\lambda$-calculus with surface reduction ($\PLambda^{\texttt{weak}}$)  and proves  using a diamond property that if  
	$\m  {\oset[-.5ex]{\surf~~~~}{\full^{k}}}   \m_k$ and $\m {\oset[-.5ex]{\surf~~~~}{\full^k}} \r_k$ (both sequence have  k steps)
	then $\forall V\in \Val$, $\mu_k(V) = \rho_k(V)$. Hence   claim (1.) follows. 
	
 	Claim (2.) follows from (1.) and from Lemma~\ref{lem:closed}, which  implies that if $M\sred \n$ is \emph{closed}, we can always choose a surface step which is a $\lred$-step.
\end{proof}

Putting all elements together, we have proved that the limit distribution of  \emph{any} $\sfull$-sequence from $\m$ is $\den \m$.
In particular, 
$\den \m$ is also the limit distribution of the  $\lfull$-sequence from $\m$.
\begin{thm} \label{thm:eval_final} For $\m$ closed, the following hold.
	\begin{enumerate}
		\item Let  $\S$ be  any $\sfull$-sequence from $\m$. Then $\S\down \den{\m}$.
		\item Let  $\L$ be  the $\lfull$-sequence from $\m$. Then $\L \down \den{\m}$.
		\item 		The  sets  		$\{\rho \st \m\tolim \rho\}, \{\rho \st \m\tolims \rho\}$, and 
		$\{\rho \st \m\toliml \rho\}$ have the same greatest element, which is $\den \m$.
	\end{enumerate}
\end{thm}

While   left reduction is not standard for finite sequences (as Example~\ref{counterex:standard} shows), still is able to reach $\den \m$, if we only evaluate programs, i.e., closed terms.
Thm.~\ref{thm:eval_final}  justifies (a posteriori!) the use of the leftmost-outermost  strategy in the literature of probabilistic $\lambda$-calculus: left evaluation actually produces  the  best asymptotic result. However,  it is not the only strategy to achieve this:  any $\sfull$-sequence will.


%


\section{Summing-up and Overview} 

The definition of reduction in $\PLambda^\cbv$ is based on two components: the $\beta_{v}$-rule and the $\oplus$-rule. We stress that  only the $\oplus$-step is constrained,     while   $\beta_{v}$ is inherited "as is" from the  $\lambda$-calculus. The $\beta_{v}$-rule is allowed in \emph{all} contexts, while the $\oplus$-rule is disabled in  a function body.
This avoids  confusion between duplicating a function which performs a choice, and duplicating the choice, that  is   the core of  confluence failure. 
It is then natural to expect that the fine control on duplication which is offered by  linear logic could be  beneficial. 

In Sec.~\ref{sec:Plinear} we apply the methods and tools which we have developed to study $\PLambda^\cbv$ to 
define a probabilistic linear calculus $\PLambda^!$ which extends with a probabilistic choice  Simpson's   linear  $\lambda$-calculus  \cite{Simpson05}. 
   This  is a result of interest in its own, but also evidence that our  approach  is robust, as it transfers well to other probabilistic calculi.
In Sec.~\ref{sec:CBN} we then define a call-by-name probabilistic calculus, $\PLambda^{\cbn}$, and we show that  similar results to the ones we have established for $\PLambda^\cbv$ hold. 

As we will see, the three calculi follow the same pattern:  the $\oplus$-reduction  (and \emph{only this} reduction) is restricted to  surface contexts. In Sec. \ref{sec:conclusion} we 
discuss how the three calculi relate.

\section{Probabilistic Linear Lambda Calculus}\label{sec:Plinear}
$\Lambda^!$ \cite{Simpson05} is an untyped linear $\lambda$-calculus which is closely based on linear logic.
Abstraction is refined into linear abstraction $\lam x.M$ and  non-linear abstraction $\lam!x.M$, which allows duplication of the argument. The argument of  $\lam!x.M$ is required  to be suspended as thunk  $!N$, that  corresponds to the $!$-box of linear logic.
In this section, we define a probabilistic linear $\lambda$-calculus $\PLambda^!$ by extending $\Lambda^!$ with an operator $\oplus$. 
\SLV{}{We demand that probabilistic choice is not reduced under the scope of a $!$ operator, while  the {$\beta$-reduction is unrestricted}.
We show that this suffices to preserve confluence; we then study the  properties of  the calculus.}

\subsection{Syntax of $\PLambda^!$}
\subsubsection{The language}
Raw \emph{terms}  $M, N, \dots$   are built up from a countable set of variables $x, y, \dots$  according to the grammar:
\begin{center}
	{\footnotesize 
	$
	\begin{array}{lcllr} M & ::= & x \mid !M \mid
	\lambda x.M   \mid 	\lambda! x.M \mid
	MM \mid	M\oplus N  & (\textbf{terms } \PLambda^!)
\end{array}
$
}
\end{center}
We say that $x$ is affine (resp. linear) in $M$ if $x$ occurs free at most (resp. exactly) once in $M$, and moreover, the free occurrence of $x$ does not lie within the scope of a $!$ operator.  A term $M$ is affine (resp. linear) if for
  every subterm $\lam x.P$ of $M$, $x$ is so in $P$.
Henceforth, we consider affine  terms only.


\SLV{}{
	It is immediate to observe that if $M$ is affine (linear) and $M\red N$, then $N$ is affine (linear).
}

\emph{Contexts} ($\cc$) and \emph{surface contexts} ($\ss$) are generated  by the grammars:
\begin{center}
	{\footnotesize  
$
\begin{array}{lclll} 
\ccontext & ::= & \square  \mid M\cc \mid \cc M   \mid \lambda x.\cc \mid  \lambda! x.\cc  \mid !\cc 
\mid \cc \oplus M \mid	M \oplus \cc &(\textbf{contexts})\\

\ss & ::=& \square \mid M\ss \mid \ss M   \mid \lambda x.\ss \mid  \lambda! x.\ss & (\textbf{surface c.})
\end{array}
$
}
\end{center}
where $\square$ denotes the \emph{hole} of the term context. 
Observe  that a \emph{surface} context is defined in a  different way than in \ref{sec:syntax}. Here it expresses the fact that  a 
\emph{surface redex cannot occur  in the scope of a $!$ operator} (nor in the scope of a $\oplus$).

\subsubsection{Reductions}We follow the same pattern as for $\PLambda^\cbv$.
The beta  rules $\mapsto_{\beta}$ are given in Fig.~\ref{fig:ruleslin}. 
The probabilistic rules $\mapsto_{l\oplus},\mapsto_{r\oplus}$
are as in  Fig.~\ref{fig:rules}. The reduction steps are in Fig.~\ref{fig:steps}:  the $\beta$-rule is closed under general context, while the $\oplus$-rules are closed under surface contexts. The $\beta$-rules   also can be restricted to  the closure under surface contexts, as shown in Fig.~\ref{fig:steps}.  A $\red$-step is \emph{deep} (written $\dred $) if  it is \emph{not} surface.
The \emph{lifting} of the relation $\red:\PLambda^!\times \MDST{\PLambda^!}$ to a binary relation  on $\Red$ $\MDST{\PLambda^!}$  is defined as in   Fig.~\ref{fig:lifting}.

\vspace{-2pt}
\begin{figure}[h]\centering
	{\scriptsize 
		\begin{tabular}{|l | l|}
			\hline
			\textbf{	Beta	 Step $\redb$}  & 	\textbf{Surface	Beta Step $\sredb$}\\[4pt]
			\infer{\cc(M) \redb \multiset{\cc(M')}}{M\mapsto_{\beta} M'} & 
			\infer{\ss(M) \sredb \multiset{\ss(M')}}{M\mapsto_{\beta} M'} \\[4pt]
			\hline \hline
	         \quad   &	\textbf{ (Surface) $\oplus$-Step  $\redo:=~\sredo$}\\[4pt]
              & 
				\infer{\ss(M\oplus N) \sred_{\oplus} \multiset{\frac{1}{2}\ss(M), \frac{1}{2}\ss(N)}}{M \oplus N\mapsto_{l\oplus} M\quad M \oplus N \mapsto_{r\oplus}N}\\
				&  \\
			\hline
			\hline
			\textbf{		Reduction Step $\red$}  & 	\textbf{Surface		Reduction Step $\sred$}\\[4pt]	
			$\red~:=~\redb \cup \redo$& $ \sred~:=~ \sredb \cup \redo$\\[4pt]
			\hline
		\end{tabular}				
		\caption{Reduction Steps }\label{fig:steps}
	}
\end{figure}
\vspace{-6pt}

\begin{figure}[h]\centering
	{\scriptsize 	
		\fbox{
			\begin{minipage}[c]{0.45\textwidth}\centering
				$
				(\lambda x.M)N \mapsto_{\beta} M[N/x] 
				\quad		\quad
				(\lambda !x.M)!N \mapsto_{\beta} M[N/x]
				$
			\end{minipage}
	}}		\caption{$\beta$ reduction rules for $\PLambda^!$ }\label{fig:ruleslin}\end{figure}

\vspace{-4pt}

\begin{remark}
	To limit   notations for reductions and contexts, 
	we use the same as for  $\PLambda^\cbv$, clearly the meaning is  different.
\end{remark}

\subsection{$\PLambda^!$ is a conservative extension of $\Lambda^!$}

As in \ref{sec:tr_cbv}, 
we  denote by $\redb$ both the reduction in $\Lambda^!$ and the $\beta$ reduction in $\PLambda^!$; we prove that   $(\PLambda^!, \Redb)$ is a conservative extension of $(\Lambda^!,\redb)$. 
\begin{Def}[Translation]\label{def:tr_lin}
$\trb{\cdot}: \PLambda^! \rightarrow \Lambda^!$ is defined in the following way, where $z$ is a fixed fresh variable
$$
\begin{array}{lcl|lcl}
\trb{x}&=&x & \trb{\lambda x. M}&=& \lambda x. \trb{M}\\
\trb{M \oplus N}&= &z ~!\trb{M} ~!\trb{N} & \trb{\lambda !x. M}&=& \lambda !x. \trb{M}\\
\trb{MN}&=&\trb{M}\trb{N}&\trb{!M}&=& ! \trb{M}\\
\end{array}
$$
\end{Def}
 Note that the translation of terms of the form $M \oplus N$ is designed so to preserves  surface  reduction.
\SLV
{\begin{prop}[Simulation]\label{prop:translationlin} 
		The translation $\trb{\cdot}$ is  sound and complete, and preserve surface reduction. 
		Let   $\red_b\in \{\redb, \sredb\}$. For $M \in \PLambda^!$, the following hold
		\begin{enumerate}
			\item $M \red_b \mset{N}$ implies $\trb{M}\red_b\trb{N}$.
			\item $\trb{M}\red_b P $ implies that exists (unique) $N\in \PLambda^!$, with $N=\trb{P}$ and $M \red_b  \mset{N}$. 
		\end{enumerate}
	\end{prop}
}
{
\begin{prop}[Simulation]\label{prop:translationlin}
Let $M \in \PLambda^!$. 
\begin{enumerate}
\item $M \redb \mset{N}$ implies $\trb{M}\redb \trb{N}$.
\item $\trb{M}\redb  P $ implies that exists (unique) $N\in \PLambda^!$, with $N=\trb{P}$ and $M \redb  \mset{N}$. 
\item $M \sredb  \mset{N}$ implies $\trb{M}\sredb  \trb{N}$.
\item $\trb{M}\sredb  P $ implies  exists (unique) $N\in \PLambda^!$, s.t.  $N=\trb{P}$ and $M \sredb \mset{N}$. 
\end{enumerate}
\end{prop}
}

The translation tells us that the reduction $\mset M\Redb \mset N$ on $\PLambda^!$ behaves as the reduction $\trb{M}\redb
\trb{N}$ on $\Lambda^!$.

\subsection{Confluence and Finitary Standardization  for  $\PLambda^!$}

The following properties hold for  $\Lambda^!$  \cite{Simpson05}.  
\begin{theorem*}[Simpson 05] The following hold in $\Lambda^!$.
\begin{enumerate}
\item \textbf{Confluence}. $\redb$ is confluent.
\item \textbf{Surface Standardization.} If $M\redb^* N$ then exists $R$ such that $M\sredb^* R$ and $R\dred^* N$.
\condinc{}{	\item $\sredbvl$ is quasi-diamond (see Sec.~\ref{sec:diamond}).}
\end{enumerate}
\end{theorem*}
We  show, using  the methods   developed for $\PLambda^\cbv$ and the translation in  Def.~\ref{def:tr_lin}, that the same properties  hold  for $\PLambda^!$.

\subsubsection{Confluence} We 
follow the same approach as  in Sec.~\ref{sec:confluence}.
In fact, we already have most of the building blocks for the proof.  Observe   that Lemma  \ref{lem:pointwise} is general enough to  apply  also  to binary relations on $\MDST{\PLambda^!}$.
\begin{lemma}\label{lem:conf_lin}
\begin{enumerate}
\item 	The  reduction $\Redo$ is diamond.
\item 	The reduction $\Redb$  is confluent.
\item 	The reductions $\Redb$ and $\Redo$ commute.
\end{enumerate}
\end{lemma}
\begin{proof}
\SLV{}{The details of the proof are in Appendix~\ref{sec:app_conflin}.}	
The proof of	 1) and 2) is  as for  Lemmas \ref{lem:oplusCR} and   \ref{lem:betaCR}; 3) is proved using   Lemma~\ref{lem:pointwise}, by induction on the term.
\end{proof}
By  Hindley-Rosen Lemma, we obtain 
\begin{thm}\label{thm:conf_lin} The reduction $\Red$ of $\PLambda^!$   is confluent.
\end{thm}

\subsubsection{Surface standardization}

\begin{prop}[Finitary Surface Standardization]\label{thm:surfacestandard_lin} In $\PLambda^!$, if $\m\Red^*\n$ then exists $\r$ such that $\m \sRed^* \r$ and $\r \dRed^* \n$.
\end{prop}
	\SLV{The proof  (see \cite{long})  builds on the analogous result for $\Lambda^!$ and is straightforward.}
	{\begin{proof}
       The proof is given  in Appendix~\ref{sec:postponement_lin}.
		\end{proof} }

\subsection{Asymptotic behaviour}

\emph{Normal forms} are defined as in  \ref{def:nf};  we denote  by $\Nnf^!$ the set of $\red$-normal forms, and by  $\Snf^!$
the set of the \emph{surface normal forms} (\ie\ the $\sred$-normal forms). Clearly $\Nnf^!  \subsetneqq \Snf^!$.
We  define $\Nnf^{!}_\sing:=\{\{M\}, M\in \Nnf^{!} \}$, 
and $\Snf^!_\sim$ as the set of all  events  $\textbf{R}:= \{S\in \Snf^! \st S=_{\beta}  R\} $.

\paragraph{Observations} 
A set of \emph{observations} for $(\PLambda^!, \Red)$ is defined in the same way as  that for  $(\PLambda, \Red)$ (Def.~\ref{def:obs} ).

\begin{prop}\label{prop:obs_lin} Each of the following sets $\set{ \Nnf^!}$, $ \set{\Snf^!}$, $\Nnf^!_{\sing}$, 
	 $\Snf^!_\sim$, is a  set  of observations for $(\PLambda^!, \Red)$.
\end{prop}


\paragraph{Limit distributions and evaluation} Once we fix a set of observations $\Obs$ for $(\PLambda^!, \Red)$,
the definition of  evaluation and limit distribution, and the notations  	$\seq{\m} \down\brho$, 	$\m\tolim \brho$ and   $\Lim(\m)$ 
are as in Def.~\ref{def:limits}.
We already observed that  Thm.~\ref{thm:unique} and \ref{thm:adequacy} only depends on confluence, and on the definition of observations;
therefore both  hold.

\begin{thm} For any choice of $\Obs$, 
	$\PLambda^!$ has the  properties:
\begin{itemize}
\item 	$\Lim(\m)$ has a greatest element, which we indicate as  $\den \m$.
\item 	If $\m \Red^* \s$, then	$\den\m = \den\s$.
\end{itemize}
\end{thm}

\paragraph{Asymptotic  Standardization} 
For the rest of the section we focus on   $\Obs:=\Snf^!_\sim$.
\SLV{}{Notice that   if  $\brho$ is a limit distribution,  $\brho\in \MDST{\Snf^!_\sim}$.}
We have established that for each  $\m\in \PLambda^!$,  $\Lim(\m)$ has a unique maximal element $\den \m$. We  now want  to 
have  a strategy to find $\den \m$. Surface reduction plays that role.
We use the following fact, which  is easy to verify.

\SLV{
\begin{fact}\label{fact:key_lin}
	Let $M\dred \n$. Then
	1) $\n$ is of the form $\mset{N}$, and $M=_{\beta} N$;
		2) $M\in \Snf^!$ if and only if $N\in \Snf^!$.
\end{fact}
}{
\begin{fact}\label{fact:key_lin}
Let $M\dred \n$. Then
\begin{enumerate}
\item $\n$ is of the form $\mset{N}$, and $M=_{\beta} N$;
\item $M\in \Snf^!$ if and only if $N\in \Snf^!$.
\end{enumerate} 
\end{fact}
}

\begin{thm}[Asymptotic Completeness]\label{thm:surfeval_lin} In $\PLambda^!$ it holds that	$\m\tolim\bmu$ if and only if   $\m \tolims \bmu$.
\end{thm}
\begin{proof}
As for    Thm.~\ref{thm:surfeval}, now using Fact~\ref{fact:key_lin} and Prop.~\ref{thm:surfacestandard_lin}.
\end{proof}
Similarly to Sec.~\ref{sec:left_evaluation}, we can establish that any (infinitary) $\sfull$-sequences from $\m$ converges precisely to $\den \m$, where $\sfull$ indicate the full lifting of  the relation $\sred \subseteq \PLambda^!\times \MDST{\PLambda^!}$.
\begin{thm}[Surface Evaluation] \label{thm:surflin}
Let  $ \S = \seq \s$ be  any  $\sfull$-sequences   from $\m$. It holds that  $\S \down \den{\m}.$
\end{thm}


\section{Call-by-Name calculus $\PLambda^{\cbn}$}\label{sec:CBN}
We  show that  results similar to those for $\PLambda^\cbv$ hold for a CbN calculus, denoted $\PLambda^\cbn$.
We could adapt all the proofs, but  now we prefer to follow a different way.
Once we take the point of view of linear logic, we  have  a roadmap to  CbN  via Girard's translation of intuitionistic  into linear logic.  More precisely, we rely on recent work \cite{EhrhardG16,gm18} which  expresses those translations  in  \emph{untyped} $\lambda$-calculus. We exploit the faithful nature of the translation   to transfer both confluence and standardization from $\PLambda^!$ to $\PLambda^\cbn$, essentially for free.

\subsection{Syntax of $\PLambda^\cbn$}
We write $\PLambda^\cbn$ for the set of terms $\PLambda$ equipped with the reduction relation $\Red$  defined below.
\subsubsection{The language} \emph{Terms and contexts} $(\cc)$ are the same as in  $\PLambda^\cbv$.
\emph{Surface  contexts} ($\ss$) are generated  by the grammar:
\begin{center}
	{\footnotesize 
$
\begin{array}{lcllr} 
\ss & ::=& \square \mid \lam x.\ss \mid \ss M    & (\cbn \textbf{ surface contexts})
\end{array}
$
}
\end{center}
\subsubsection{Reductions} 
The $\beta$-rule $\mapsto_{\beta}$  is as in the CbN $\lambda$-calculus (Fig.~\ref{fig:rulescbn}). The probabilistic rules $\mapsto_{l\oplus},\mapsto_{l\oplus}$
are as in  Fig.~\ref{fig:rules}. 

\vspace{.1in}
\begin{figure}[h]\centering
	{\scriptsize 	
		\fbox{
			\begin{minipage}[c]{0.45\textwidth}\centering
				$
				(\lambda x.M)N \mapsto_{\beta} M[N/x] 
				$
			\end{minipage}
	}}		\caption{  Beta Reduction Rule for $\PLambda^\cbn$}\label{fig:rulescbn}\end{figure}
\vspace{.1in}

\emph{Reduction steps} $\red, \redb,\redo\subseteq \PLambda\times \MDST{\PLambda}$ and \emph{surface reduction steps}
$\sred, \sredo,\sredb  \subseteq \PLambda\times \MDST{\PLambda}$  are defined 
 in Fig.~\ref{fig:steps},  following the usual pattern. 
By  definition of  surface context, a reduction step is \emph{surface} if  it \emph{does not occur in argument position} (nor in the scope of  $\oplus$). 

\condinc{}{A reduction step is \emph{deep} (written $\dred $) if  it is \emph{not} surface.}

The \emph{lifting} of $\red\subseteq\PLambda\times \MDST{\PLambda}$ to a binary relation $\Red$ on $\MDST{\PLambda^\cbn}$  is defined as in Fig.~\ref{fig:lifting}. The full lifting $\full$  is defined as in \ref{def:fulllift}.

%
%

\subsubsection{Normal Forms } \label{sub:normal}We denote  by $\Nnf^\cbn$ the set of $\red$-normal forms, and by  $\Snf^\cbn$
the set of the \emph{surface normal forms} (\ie\ the $\sred$-normal forms). Clearly $\Nnf^\cbn  \subsetneqq \Snf^\cbn$.

Let us extend to $\PLambda^{\cbn}$ the notion of \emph{head normal form}.
 Head reduction $\hred$ is 
the closure of both the $\beta$ and the probabilistic rules  under head context $\hh$,
which is defined by the following grammar
\begin{center}
	{\footnotesize $\hh  ::= \lam x.\hh \mid  \kk   \quad\quad \kk::= \square \mid  \kk M \quad\quad (\textbf{ head contexts })$}
\end{center}

\SLV{}
{\begin{remark*}A common way to write head context $\hh$ is as follows:
	{\footnotesize $$\hh  ::= \lam x_1 \dots \lam x_k.\square P_1\dots P_n \quad \quad (\textbf{ head contexts })$$}
\end{remark*} }

Observe  that $\hred ~\subsetneqq ~ \sred$ (for example, the reduction  $(\lam x. (\lam y.y) P) Q \sred (\lam x. P) Q  $    is not a head reduction). However, the two relations have the same normal forms. Let us  write  $\Hnf$ for the set of \emph{head normal forms} . 
\SLV{}{If $M$ is in  surface normal form, it is also in head normal form.} 
It is easy to verify that a head normal form has no  $\sred$-redex, and   conclude:
\begin{center}
$\Snf^\cbn = \Hnf$
\end{center}
\subsubsection{$\PLambda^{\cbn}$ to $\PLambda^{!}$.}\label{sec:tr_cbn}
 In  \cite{gm18},  the  translation from  $\Lambda^\cbn$ into a linear $\lambda$-calculus  is  proved sound and complete.
We follow their work to define a similar  translation  $\trn{\cdot}:\PLambda^\cbn \to \PLambda^!$:
\[
\begin{array}{lcl|lcl}
\trn x&= &x &  \trn{\lam x.M}&=&\lam !x.\trn{M}\\
\trn{MN} &=&  \trn{M}!\trn{N} &
\trn{M\oplus N} &=& \trn{M}\oplus \trn{N}\\
\end{array}
\]
\[
\begin{array}{lcl}
 \trn{\mset{p_iM_i \st \iI}} &=& \mset{p_i\trn{M_i} \st \iI} 
\end{array}
\]
The following   extend to the probabilistic setting an  analogous result proved in \cite{gm18}.
Observe that,  with a slight abuse of  notation,  reductions in the two calculi are denoted in the same way, the meaning being clear from the context.
%
%


\begin{prop}[Simulation]\label{thm:ntr} The  translation $\trn{.}$ is sound and complete; it preserves surface reduction and surface normal forms.
\SLV{
		Let   $\red_b\in \{\redb, \sredb\}$. For	
		 $M\in \PLambda^\cbn$, it holds that:
		\begin{enumerate}
			\item if $M \red_b \n$  then $\trn{M}\red_b\trn \n$ ;
			
			\item if $\trn{M}\red_b \s$ then $\exists ! \n$  such that $\s= \trn \n$ and $M\red_b \n$;

			\item $M\in \Hnf$ if and only if $\trn{M} \in \Snf^!$.
	\end{enumerate} 	
}
{	
Let $M\in \PLambda^\cbn$; the following hold:
	\begin{enumerate}
	\item if $M \red \n$  then $\trn{M}\red\trn \n$;
	\item if $M \sred \n$  then $\trn{M}\sred\trn \n$;
	
	\item if $\trn{M}\red \s$ then $\exists ! \n$  such that $\s= \trn \n$ and $M\red \n$;
	\item if $\trn{M}\sred \s$ then $\exists ! \n$  such that $\s= \trn \n$ and $M\sred \n$;
	
	\item $M\in \Hnf$ if and only if $\trn{M} \in \Snf^!$.
\end{enumerate} 	}
\end{prop}
\begin{proof}
\SLV{The proof is in \cite{long}.}{	The proof is in  Appendix~\ref{sec:tr_cbn}.}
\end{proof}

\subsection{Confluence and Finitary  Standardization  for $\PLambda^\cbn$ }
  The fact that surface reduction is preserved by $\trn{.}$  is crucial to transfer the standardization result from $\PLambda^!$ to $\PLambda^\cbn$.  We show that via translation, 
$\PLambda^{\cbn}$ inherits both the confluence and  the surface standardization property from $\PLambda^!$.

\begin{thm}[Confluence]
The relation $\Red_\cbn$ is confluent.
\end{thm}
\begin{proof}
From Thm.~\ref{thm:conf_lin},  using back-and-forth  Thm~\ref{thm:ntr}.
\end{proof}

\begin{thm}[Finitary Surface standardization]\label{thm:surface_cbn} 
If $\m\Red^* \n$ then exists $\r$ such that $\m\sRed^* \r$ and $\r\dRed^* \n$.
\end{thm}
\begin{proof}
From Thm.~\ref{thm:surfacestandard_lin}, by using  back-and-forth Thm~\ref{thm:ntr}, and the fact that the translation preserves surface reduction. 
\end{proof}

In the classical $\lambda$-calculus, the standardization property (Barendregt, Th.~11.4.7) says that every reduction sequence can be ordered in such a way to perform first only  left $\beta$-redexes, reading the term from left to right, and then  internal ones (a redex is internal if it is not the leftmost one).  

In $\PLambda^\cbn$ this  notion of standardization fails, as the following example (which we take from \cite{LeventisThesis}) shows.
\begin{example}\label{ex:counter2}
	In each step, we underline the redex. Consider 
	$\mset{(\lam x. \underline{I (y\oplus z)}) I} \Red  \mset{ (\lam x. \underline{y\oplus z}) I} \Red \mset{\underline{{\two}(\lam x. y) I}, {\two}(\lam x.  z) I}   \Red \mset{{\two}y, {\two}(\lam x.  z) I}$,
	where only the last step reduces a left redex.
	If we perform the left redex first, we have  $\mset{(\lam x.I (y\oplus z)) I}\Red \mset{I(y\oplus z)}$, from which 
	$\mset{{\two}y, {\two}(\lam x.  z) I}$ cannot be reached. 
\end{example}

A consequence of standardization 
 is that $M$ has a head normal form iff the $\hred$-sequence from $M$ terminates. 
In the following section we  retrieve an analogue of this result.

\subsection{Asymptotic behaviour}
We denote by $\Hnf_\sim$  the  set of  head normal forms  up to the equivalence $=_{\beta}$, and we define $\Nnf^{\cbn}_\sing=\{\{M\}, M\in \Nnf^{\cbn} \}$.	

\paragraph{Observations} Observations are defined as in Def.~\ref{def:obs}.

\begin{prop}
Each of the following  is a set of observations for $\PLambda^\cbn$:  $\{\Hnf\}$, $\{\Nnf^\cbn\}$,  $\Hnf_\sim $, $ \Nnf^\cbn_\sing $.
\end{prop}


\paragraph{Convergence and  Limit distributions} Once we fix a set of observations $\Obs$ for $\PLambda^\cbn$,
the definition of  convergence and limit distribution
are as in Def.~\ref{def:limits}.
We observe that  Theorems~\ref{thm:unique} and \ref{thm:adequacy} both hold. Hence in particular
\begin{thm}For any choice of $\Obs$, the following holds in 
	$\PLambda^\cbn$ : given $\m$, 
	$\Lim(\m)$ has a greatest element  $\den \m$.
\end{thm}


\condinc{}{
\begin{thm} With the choices and notations above,  and for any choice of $\Obs$ as in Prop.~\ref{prop:obs_cbn}$,
	\PLambda^\cbn$ has the following properties:
	\begin{itemize}
		\item 	$\Lim(\m)$ has a greatest element, which we indicate as  $\den \m$.
		\item 	If $\m \Red^* \s$, then	$\den\m = \den\s$.
	\end{itemize}
\end{thm}
}

We now study the notion of  convergence induced by\emph{ choosing head normal forms as outcome}, \ie\ $\Obs:=\Hnf_\sim$.
Therefore,  if $\brho \in \Lim(\m)$, it holds  $\brho\in \MDST{\Hnf_\sim}$.
The following results match the analogous results in $\PLambda^!$ (Thm.~\ref{thm:surfeval_lin} and \ref{thm:surflin}).
\begin{thm} Let $\Obs:=\Hnf_\sim$. For every multidistribution $\m$:
\begin{itemize}
\item $\m\tolim\bmu$ if and only if   $\m \tolims \bmu$.
\item If  $\seq \s$ is a $\sfull$-sequences of full surface reductions from $\m$, then  $\seq \s \down \den{\m}.$
\end{itemize}
\end{thm}

Similarly to Prop.~\ref{thm:lefteval}, it is not hard to 
  prove that in $\PLambda^\cbn$, $\sfull$ satisfies a diamond property 
in the sense of \cite{pars}, and hence 
 all  $\sfull$-sequences  from $\m$  
 converge to the same limit distribution.
 \SLV{}{Since $\hred \subset \sred$ and  since head reduction  and surface reduction have the same normal forms, we can always  choose a $\lred$ step whenever a $\sred$-step is possible.  This allows us to retrieve a result of completeness for head reduction:}
\emph{\begin{center}
Let   $\seq \s$ be  \emph{the}  $\hfull$-sequences of full head reductions from $\m$. It holds that  $\seq \s \down \den{\m}.$
\end{center}}
Once again, this justifies a posteriori the choice of head reduction in probabilistic CbN (such as \cite{EhrhardPT11}).
Observe that  we follow   the same reasoning as in the case of $\PLambda^\cbv$ (with  $\Val_{\sim}$ as  set of observations). First we proved that surface reduction is sufficient to reach the greatest limit distribution, then we observed that in particular left reduction can be chosen.  There is a close parallelism between $\PLambda^\cbv$ and $\PLambda^\cbn$: similar results hold if we consider as set of observations $\Val_{\sim}$ and $\Hnf_{\sim}$ respectively.

\section{Conclusion and discussion}\label{sec:conclusion}
\subsection{Summary}
In this paper we  design  two  probabilistic extensions of respectively  the CbV and CbN $\lambda$-calculus, $\PLambda^{\cbv}$ and $\PLambda^{\cbn}$,   which we propose as foundational calculi for  probabilistic computation. Both  calculi enjoy \emph{confluence} and \emph{standardization, in an extended way}. Namely, first we prove both properties  for the finite  sequences, exploiting classical methods,  then we extend these properties to  the limit, developing new sophisticated proof methods. 
In particular, we prove the \emph{uniqueness} of the (maximal) result, parametrized by the notion of \emph{set of observations},
and  that  the  asymptotic extension of surface standardization supplies a family of \emph{complete reduction strategies} which are guaranteed  to reach  the best  result.
The two calculi have a common root in the linear $\lambda$-calculus  $\PLambda^!$, which is both a technical tool
and  a calculus of interest in its own, in which  a fine control of the interaction between copying and choice is possible.

{In all three calculi, $\beta$-reduction is unconstrained; hence for each calculus, its  restriction  to only $\beta$-reduction exactly gives  the  usual corresponding (CbN, CbV, or linear) $\lambda$-calculus; this is not the case for extensions in which a  strategy is fixed.
}

New proof methods include  the \emph{asymptotic extension of surface standardization} (Thm.~\ref{thm:surfeval}), and the use of a \emph{translation to transfer standardization} properties, namely  from $\PLambda^!$ to  $\PLambda^{\cbn}$.
It is  worth stressing  a crucial  element: the fact that the  translation is  sound, complete and {preserves surface contexts} is what allows us to transfer the results. 

%
%

\subsection{Discussion}
\paragraph{Relating the  calculi (Girard's Translations)}\label{sec:bigpic}
The key to understand how $\PLambda^\cbv$, $\PLambda^\cbn$, and  $\PLambda^!$ relate are 
the two Girard's  translations which embed intuitionistic logic into linear logic, and which are well known to respectively  correspond to   CbN and  CbV  computations. Let us clarify this.
Let us  start from   $\PLambda^!$: the natural constraint to avoid copying the result of a choice is  "no $\oplus$-reduction  in the scope of $!$" (\ie, inside a !-box). 
Using the intuition provided by Girard's translations as a guide, the constraint  above becomes respectively  
"no $\oplus$-reduction in the scope of a $\lambda$-abstraction" (in  CbV) and  "no $\oplus$-reduction in argument position" (in  CbN).
Our three notions of surface context 
express  these three constraints.


The intuitive reasoning above can  be formalized  thanks to a recent line of work \cite{EhrhardG16,gm18}, which internalizes the insights coming from linear logic and proof nets into a $\lambda$-syntax.  The resulting  calculus  subsumes both CbN and CbV $\lambda$-calculi via Girard's translation.
The idea of a system which subsumes both CbV and CbN   had been already advocated and developed  by Levy, via the Call-By-Push-Value paradigm \cite{Levy99}.
And indeed, \cite{EhrhardG16} can be seen as an untyped version of Levy's calculus. We  leave to the future a comprehensive approach, where a probabilistic linear calculus is the  metalanguage in which all the results are developed.

\paragraph{On non-deterministic $\lambda$-calculi}
The finitary results we presented (namely, confluence and finitary surface standardization)  also hold if the probabilistic choice is replaced by  \emph{non-deterministic} choice (just forget the coefficients). Asymptotic results, instead,  are specific to probabilistic computation.

	\paragraph{$\PLambda^!$ and quantum  $\lambda$-calculi}
	The fine control of duplication which  $\Lambda^!$ inherits from linear logic has made it an ideal base for quantum $\lambda$-calculi (such as \cite{DalLagoMZ11,popl17}).
	In those calculi,\emph{ surface reduction} is the key ingredient to  
	allow for  the coexistence of quantum bits with duplication and erasing. 
	\SLV{}{}
	{\emph{No reduction} (not even $\beta$) is allowed in the scope of a $!$ operator.}
	Our results show that $\beta$-reduction can be unrestricted, only measurement (the quantum analogue of $\oplus$) needs to be surface.

%

\appendix

\bibliographystyle{abbrv}
\bibliography{biblioPARS}

\SLV{}{
\newpage

\appendix

\subsection{Proofs of Section~\ref{sec:finitary_stCBV} }
We prove Thm.~\ref{thm:surfacestandard}, \ie\ finitary Surface Standardization for  $\PLambda^\cbv$. We start by   establishing  Surface Standardization for the (non probabilistic) call-by-value $\lambda$-calculus, $\Lambda^\cbv$, in  \ref{sec:cbvfolklore}. This result is folklore, but we could not find it in the literature. In \ref{sec:postponement_cbv}  we extend the result to $\PLambda^\cbv$.

\subsubsection{Preliminary definitions}
\paragraph{Surface and left reduction}
Surface and left reduction have been defined in Sec.~\ref{sec:strategies}; 
Fig.~\ref{fig:sred} and Fig.~\ref{fig:lred}	 give explicitly the  inference rules for  surface   and left  steps;  we use the notation defined below:

\begin{notation*}\label{not:app} 
	If $\m= \mdist{p_iM_i\st \iI}$,  we write    $ \m @ Q$ for  $ \mdist{p_i (M_i Q)\st \iI}$, and  $Q @\m$ for $\mdist{p_i(Q M_i)\st \iI}$.
\end{notation*}
	We recall that  a reduction step $\red$ is \emph{deep}, written $\dred$, (resp. \emph{internal}, written $\ired $) if it is not  a surface step (a left step). We have already observed  that  $\dred \subset \ired$, and 
that since a $\oplus$-redex is always surface, a $\dred$ step is always a $\redbv$ step.

\begin{figure}\centering
{\scriptsize 	
\fbox{
	\begin{minipage}[c]{0.48\textwidth}
		$$\infer{(\lam x. M)V\sred  \mdist{ M[V/x]}}{V \in \Val} \quad \infer{M\oplus N \sred \mdist{\frac{1}{2}M,
				\frac{1}{2}   N}}{}
		$$
		$$\infer{MN\sred  \m @N}{M\sred \m} \quad 
		\infer{MN\sred  N @\n} { N\sred  \n}
		$$
		\caption{Surface Reduction }\label{fig:sred}
	\end{minipage}
}}
\end{figure}

\begin{figure}
	{\scriptsize 	
		\fbox{
			\begin{minipage}[c]{0.45\textwidth}
				$$\infer{(\lam x. M)V\lred  \mdist{ M[V/x]}}{V \in \Val } \quad \infer{M\oplus N \lred \mdist{\frac{1}{2}M,
						\frac{1}{2}   N}}{}
				$$
				$$\infer{MN\lred  \m @N}{M\lred \m} \quad 
				\infer{VN\lred  V @\n}{V \in \Val & N\lred  \n}
				$$
				\caption{Left Evaluation }\label{fig:lred}
		\end{minipage}}
	}
\end{figure}

\paragraph{Parallel $\beta_{v}$-reduction}
\begin{itemize}
\item \emph{Parallel} $\beta_{v}$-reduction is  a standard definition, and is given in Fig.~\ref{fig:parredb}. We define its   lifting  $\parRed_{\beta_{v}}$ as usual (see Section \ref{sec:syntax}).
\item \emph{Deep parallel} reduction	($\dparred$, with lifting $\dparRed$) 
indicates that $M\parRed \mset{S}$ and $M {\dRed}^* \mset{S}$. We  make  the rules explicit in  Fig.~\ref{fig:dparred}. 
\end{itemize}

\begin{figure}\centering

\fbox{
	\begin{minipage}[c]{0.48\textwidth}
		\[ x\parredbv \mdist{x}  \quad\quad   \infer{\lam x.M\parredbv \mdist{\lam x. N} }  {M\parredbv\mdist{N} } \]
		\[\infer{MN\parredbv \mdist{M'N'} }{M\parredbv \mdist{M'} & N\parredbv \mdist{N'}}\]
\[	\infer{(\lambda x. M)W\parredbv \mdist{M'[W'/x]} }{M\parredbv \mdist{M'} & W\parredbv \mdist{W'}  & W \texttt{ value}} 
		\]
		\[\infer{M\oplus N\parredbv \mdist{M'\oplus N'} }{M\parredbv \mdist{M'} & N\parredbv \mdist{N'}}\]
		\caption{$\beta$-Parallel Reduction }\label{fig:parredb}
	\end{minipage}
}
\fbox{
	\begin{minipage}[c]{0.48\textwidth}
		$$ x\dparredbv \mdist{x}  \quad \quad  \infer{\lam x.M\dparredbv\mdist{\lam x. N} }  {M\parredbv\mdist{N} } $$
		$$\infer{M N\dparredbv  \mdist{S T}}{ M\dparredbv \mdist{S} & N\dparredbv \mdist{T} }$$
		$$\infer{M\oplus N\dparredbv  \mdist{S \oplus T}}{ M\parredbv \mdist{S} & N\parredbv \mdist{T} }$$
		\caption{Deep Parallel Reduction}\label{fig:dparred}
	\end{minipage}
}
\end{figure}

\begin{fact}\label{prop:trans}  The following holds
$$\dRed_{\beta_v} ~~ \subseteq ~~ \dparRed_{\beta_v} ~~\subseteq ~~\dRed_{\beta_v}^*$$
\end{fact}

\paragraph{Translation}
\newcommand{\z}{z}
\newcommand{\lamred} {\redbv}   
\newcommand{\parlamred}{\parred}
We refine the translation given in in Sec.~\ref{sec:tr_cbv} in order to \emph{preserves surface reduction}. 

Let $z,w$ be fresh variables.
	$(\cdot)_{\lambda}: \Lambda_{\oplus} \rightarrow \Lambda$ is    defined as follows:

$$
\begin{array}{lcl|lcl}
(x)_{\lambda}&=&x & 	(MN)_{\lambda}&=&(M)_{\lambda }(N)_{\lambda}\\
 \tr{M\oplus N} & = &  z(\lam w.\tr M)\lam w.\tr N &  \tr{\lambda x. M} &=& \lambda x. \tr{M}\\
\end{array}
$$

\vskip 8pt
The following is straightforward to check.
\begin{lemma}\label{simulation}
Assume $M\in \PLambda$.

\begin{enumerate}
\item $P\redbv \mdist{Q} $ and  $\tr Q = S$ (in $\PLambda$) $\Longleftrightarrow $     $\tr P \lamred S$ (in $\Lambda$).  
\item $P\sredbv \mdist{Q} $ and  $\tr Q = S$ (in $\PLambda$) $\Longleftrightarrow $     $\tr P \sred_{\beta_v} S$ (in $\Lambda$).  
\item  $P\parredbv \mdist{Q} $ and   $\tr Q = S$  (in $\PLambda$) $\Longleftrightarrow $       $\tr P \parlamred_{\beta_v} \tr Q$  (in $\Lambda$).  

\end{enumerate}
\end{lemma}

\condinc{}{
\begin{proof}Immediate. Let us explicitely $\Leftarrow$, for example in Case 2,  by induction on $P$. We only discuss the case $P=M\oplus N$.
Assume $\tr {(M\oplus N)} \parlamred Q$. Then $Q=X\z_pY$ and the derivation must be as follows
\[\infer{ \tr M \z_p \tr N\parlamred X \z_p Y}{\infer{\tr M\z_p\parlamred  X\z_p}{\z_p\parlamred \z_p & \tr M\parlamred X} & \tr N \parlamred Y}\] 

By induction, there are terms $S,T\in \PLambda$ such that  $X=\tr S$, $Y=\tr T$, $M\parredbv \mdist{S}$ and $N\parredbv \mdist{T}$. We therefore conclude that $M\oplus N \parredbv \mdist{R}$, with $\tr R= X\z_pY$	
\end{proof}
}



\subsubsection{$\Lambda^\cbv$ and Surface Standardization}\label{sec:cbvfolklore}
With the standard definition of left, internal, and parallel reduction (denoted $\lredbv,\iredbv,\parredbv$, respectively) the following results are well known to hold (see \cite{PlotkinCbV, RonchiPaolini}).
\begin{itemize}
\item[(a)] If $M\redbv^* N$ then  exists $ S $ such that  $$M\lredbv^*S\iredbv^* N.$$
\item[(b)] If $M\parred_{\beta_v} N$ then  exists $ S $ s.t.  $M\lredbv^*S\iparred_{\beta_v} N$.
\item[(c)] If $M\iparred_{\beta_v} M' \lredbv N$, it exists $S$ s.t. $M\lredbv^* S\iparred_{\beta_v} N$.
\end{itemize}


The results of the following lemma are immediately obtained from the previous ones, by observing that a left reduction is a surface reduction, a deep reduction is always an internal reduction, and that $\iredbv$ does not modify the shape of a term (see  \cite{RonchiPaolini}).
\begin{lemma}\label{lem:basicst}
\begin{enumerate}
\item If $M\redbv^* N$ then  exists $ S $ such that   $$M\sredbv^*S\dredbv^* N.$$	
\item If $M\parredbv N$ then  exists $ S $ s.t. $M\sredbv^*S\dparredbv N$.
\item If $M\dparredbv M'\sredbv N$ then exists $ S $ s.t. $M\sredbv^* S\dparredbv N$.
\end{enumerate}
\end{lemma}

\begin{proof}
	The first two are by induction on $N$. We recall that $\lred \subseteq \sred$.
	\begin{enumerate} 
		\item  By (a) $M\redbv^* N$ implies $M\lredbv^*S\iredbv^* N$. We  examine $N$.
		\begin{itemize}
			\item $N=x$. Then $S=N$ and the result holds trivially.
			\item $N=\lam x.P$. Hence  $S =\lambda x.Q \iredbv^* \lambda x.P$. 
			Then  $M \lredbv^* \lambda x.Q \dredbv^* \lambda x.P$.
			\item $N=PQ$. Then $S=P'Q'$, where $P' \redbv^* P$ and $Q' \redbv^* Q$.
			By induction $P' \sredbv^* P'' \dredbv^* P$ and $Q' \sredbv^* Q'' \dredbv^* Q$, and the desired sequence  is
			$M \lredbv^* P'Q' \sredbv^* P''Q'' \dredbv^* PQ$. 
		\end{itemize}
		  The result follows since $\lredbv \subset \dredbv$.
		\item Similar to the previous one, using (b) , i.e., the fact that  $M\parred_{\beta_v} N$ implies $M\lredbv^*S\iparred_{\beta_v} N$.
			\begin{itemize}
			\item $N=\lam x.P$. Then  $S =\lambda x.Q$ and $Q \parredbv P$. By 
			definition,  $M \lredbv^* \lambda x.Q \dparred_{\beta_v}\lambda x.P$.
			\item $N=PQ$. Then $S=P'Q'$, with  $P' \parredbv P$ and $Q' \parredbv Q$.
			By induction $P' \sredbv^* P'' \dparredbv P$ and $Q' \sredbv^* Q'' \dparredbv Q$, and the desired sequence  is
			$M \lredbv^* P'Q' \sredbv^* P''Q'' \dredbv^* PQ$.
		\end{itemize}

		\item By induction on $M$.
		\begin{itemize}
			\item $M=x$ or $M=\lam x.P$.  Immediate.
			\item $M=(\lam x.P)V$. Assume $(\lam x.P)V\dparredbv (\lam x.P')V'\sredbv N$. Since the deep step is an internal step, the surface step is a left step, we have 
			$(\lam x.P)V\iparred_{\beta_v} (\lam x.P')V'\lredbv N$. From (c),  it exists $S$, $M\lredbv^*  S\iparred_{\beta_v} N$.
			The $\iparred_{\beta_v}$ step is in particular a $\parred_{\beta_v}$, hence  from point 2 it holds that  $S\sredbv^*S'\dparredbv N$, hence the claim.
			\item $M=PQ$. By hypothesis, $PQ\dparredbv P'Q'\sred_{\beta_v} N$; the surface redex is inside either $P'$ or $Q'$, say $Q'$. We have  $N=P'R$, $Q\dparredbv Q'\sredbv R$ and by induction
			$Q\sredbv^* R'\dparredbv R$. Hence $PQ\sredbv^* PR'\dparredbv P'R$.
		\end{itemize}
	\end{enumerate}
\end{proof}

\condinc{}{
\begin{proof}
The first two are by induction of the shape of $N$. 
\begin{itemize} 
\item[1']  $M\redbv^* N$ implies $M\lredbv^*S\iredbv^* N$. Let examine $N$.
\begin{itemize}
	\item $N=x$. Then $S=N$ and the result holds trivially.
	\item $N=\lam x.P$. Then also $S =\lambda x.Q \iredbv^* \lambda x.P$. By 
	definition,  $M \sredbv^* \lambda x.Q \dredbv^* \lambda x.P$.
	\item $N=PQ$. Then $S=P'Q'$, where $P' \redbv^* P$ and $Q' \redbv^* Q$.
	By induction $P' \sredbv^* P'' \dredbv^* P$ and $Q' \sredbv^* Q'' \dredbv^* Q$, and the desired sequence  is
	$M \sredbv^* P'Q' \sredbv^* P''Q'' \dredbv^* PQ$.
\end{itemize}

\item[2'] By induction on the definition of $\parredbv$.
If $M \parredbv N$ since $M=\lambda x.M'$, $N=\lambda x. N'$ and $M' \parredbv N'$, then the proof come from induction and the definition of deep.\\
If $M \parredbv N$ since $M=PQ$, $N=P'Q'$ and $P \parredbv P'$ , $Q \parredbv Q'$, then by induction
$P \sredbv^*P''\dparredbv P'$ and $Q \sredbv^*Q''\dparredbv Q'$, then 
$PQ \sredbv^* P'' Q'' \dparredbv P'Q'$.\\
If $M \parredbv N$ since $M=(\lambda x.P)V \parredbv P'[V'/x]$, where $P \parredbv P'$ and $V \parredbv V'$,
then by induction $P \sredbv^* P'' \dparredbv P'$ and $V \sredbv^* V'' \dparredbv V'$. But in a value, all reductions are deep, so $V \dparredbv V'$. Then 
$M=(\lambda x.P)V \sredbv P[V/x]  \sredbv^* P''[V/x]  \dparredbv P'[V'/x] $.
\item[3']  By induction on the shape of $M$.
\begin{itemize}
\item $M=x$ or $M=\lam x.P$ immediate
\item $M=(\lam x.P)V$. Assume $(\lam x.P)V\dparredbv (\lam x.P')V'\sredbv N$ The deep step is an internal step, the surface step is a left step.
Therefore $(\lam x.P)V\iparred_{\beta_v} (\lam x.P')V'\lredbv N$. From (3) exists $S$, $M\lredbv^*  S\iparred_{\beta_v} N$ and from (2') $S\sredbv^*S'\dparredbv N$.
\item $M=PQ\dparredbv P'Q'$, and the surface redex is inside either $P'$ or $Q'$, say $Q'$. We have  $N=P'R$, $Q\dparredbv Q'\sredbv R$ and by induction
$Q\sredbv^* R'\dparredbv R$. Hence $PQ\sredbv^* PR'\dparredbv P'R$.
\end{itemize}
\end{itemize}
\end{proof}
}

\subsubsection{Surface Standardization in $\PLambda^\cbv$}\label{sec:postponement}\label{sec:postponement_cbv}
In order to prove Theorem ~\ref{thm:surfacestandard}, we need a lemma.

\begin{lemma}\label{lem:post}If $M\dparred \mset{M'}$ and $M' \sred \n$, then it exists $\s$, such that  $\mset{M} \sRed^* \s$ and $\s \dparRed \n$.
\end{lemma}

\begin{proof}If $M' \sredbv \n$, the claim holds by simulation in $\Lambda^\cbv$ and Lemma~\ref{lem:basicst}, point (3).
If $M' \sred_{\oplus} \n$, we procede by induction on $M$. 
\begin{enumerate}
\item The case $M=x$ and $M=\lam x.P$ do not apply.
\item Let $M=P \oplus Q$.  Assume $P \oplus Q \dparred \mset{R \oplus S}$, so $P\parredbv R$, $Q\parredbv S$, and $R \oplus S\sredo \mset{\two R,\two S}$. By  Lemma~\ref{lem:basicst}, point 2 and  simulation  in $\Lambda^\cbv$, it holds that $P\sred^* P'\dparred R$  and $Q\sred^* Q'\dparred S$. Therefore  $P \oplus Q \sredo \mset{\two P, \two Q} \sRed^*\mset{\two P' ,\two Q'}\dparRed \mset{\two R, \two S}$
\item Let $M=PQ$. Assume  $PQ\dparred \mset{RT}$ (with  $P\dparred R$ and $Q\dparred T$)   and $RT\red_{\oplus} \n$ with the $\oplus$-redex in either $R$ or $T$, say is in $R$. Hence $R\redo \r=\mset{\two R_i\mid i\in \{1,2\}} $ and $RT\sred \mset{\two R_iT\mid  i\in \{1,2\}}=\n$. 
By induction, from $P\dparred \mset{R}\sRed \r$ we have $\mset{P}\sRed^* \mset{\two S_i\mid  i\in \{1,2\}}$  and $ S_i\dparred R_i$ .
Therefore $\mset{PQ}\sred^* \mset{\two S_iQ\mid  i\in \{1,2\}} \dparRed \mset{\two R_iT\mid  i\in \{1,2\}}=\n$ . 
\end{enumerate}

\end{proof}

\begin{cor}\label{cor:post}
If $\m\dparRed \n$ and $\n \sRed^* \r$, then exists $\s$ with $\m \sRed^* \s$ and $\s \dparRed \r$.
\end{cor}
\begin{proof}
By induction on the length $k$ of $\n \sRed^{(k)} \r$. If $k=0$ the result is trivial.
Otherwise, let $\n \sRed^* \r$ be $\n\sRed \n_1 \sRed^{(k-1)} \r$. By Lemma~\ref{lem:post}, from   $\m\dparred  \n\sRed \n_1$ we have  that $\m\sRed^* \s\dparRed \n_1 \sRed^{(k-1)} \r$.
By inductive hypothesis, $\s \sRed^* r' \dparRed r$, hence $\m\sRed^* \s \sRed^* \r' \dparRed \r$.
\end{proof}

Now we are able to prove the theorem: 
\begin{center}
\textbf{Thm.~\ref{thm:surfacestandard}}:	if $\m\Red^*\n$ then  then exists $\r$ such that  $\m \sRed^* \r$ and $\r \dRed^* \n$. 
\end{center}
\begin{proof}
By induction on the length $k$ of the reduction $\m\Red^*\n$,  using  Corollary \ref{cor:post}. 


If $k=0$, the result is trivial ($\r=\m$). Otherwise, $\m\Red\m_1\Red^*\n$. By induction, we have 
$\m_1\sRed^*\r\dRed^*\n$.  We can separate the first step in two: $\m\sRed\m'\dRed \m_1$, by reducing first only the elements of $\m$ which have a surface reduction, and then only the elements which have a deep reduction. The   step $\m'\dRed \m_1$ can be regarded as a parallel step. By Corollary  \ref{cor:post}, from $\m'\dparRed \m_1\sRed^*\r$ we obtain $\m'\sRed^*\s \dparRed\r$, hence it holds that 
$\m\sRed\m' \sRed^*\s \dRed^*\r \dRed^*\n$.

\end{proof}

\condinc{}{
\subsection{Controesempi alla standardizzazione basata su  INTERNAL POSTPONEMENT...}
La claim "$\m\Red^*\n$ then $\m\lRed^*s\iRed^*\n$" non vale
con nessuna scelta per i redex  $\oplus$. La definizione di redex principale e interno di una riduzione $\redbv$ e' chiara, quella standard. 
Resta da scegliere come classificare i  redex  $\oplus$.  
In particolare, un redex
$\oplus$ che non e' leftmost (in CbV sense) deve essere considerato interno, o principale? Or is neither, and we reduce a $\oplus$-redex  whenever is possible?  
Abbiamo controesempi per entrambe le  scelte:
\begin{enumerate}
\item an $\oplus$-redex is principal if it is the "leftmost" redex (in the CbV sense), otherwise it is internal
$$ (II)(P\oplus Q)\ired_{\oplus} \mset{\two(II)P,\two (II)Q}\lredbv \mset{\two IP,\two (II)Q} $$
If we postpone $\redo$, the result is not the same
\item  let us decide that we reduce $\oplus$-redex whenever is possible, 
$$(II)((\lam x.P\oplus Q)V)\iredbv (II)P'\oplus Q'\lred_{\oplus} \mset{\two (II)P', \two (II)Q'} \lredbv \mset{\two IP,\two (II)Q}$$
We cannot postpone the $\iredbv$ reduction.
\end{enumerate}
}

\subsection{Proofs of Section~\ref{sec:asymptotic}}

\paragraph{Monotone Convergence.}

We recall the following standard result.
\begin{theorem*}[Monotone Convergence for Sums]
	Let $\mathcal X$ be a countable set,  $f_n: \X\to[0, \infty]$ a non-decreasing sequence of  functions, such that 
	$f(x):= \lim_{n\to\infty} f_n(x)=\sup_n  f_n(x) $ exists for each $x\in \X$. Then 
	$$\lim_{n\to \infty} \sum_{x\in \X} f_n (x) ~=~  \sum_{x\in \mathcal X} f(x)$$
\end{theorem*}

\condinc{}{
\begin{proof}For indexes $m<n$, we have , for each $x\in \X$
		$$ f_m(x) \leq f_n(x) \leq f(x)$$	
		Hence  $ \sum_{k:1}^\infty f_m(x_k) \leq \sum_{k:1}^\infty f_n(x_k) \leq \sum_{k:1}^\infty  f(x_k)$ and therefore
		$$L=\lim_n \sum_{k:1}^\infty f_n(x_k)  \leq  \sum_{k:1}^\infty f(x_k)$$
		
		To prove the opposite inequality,  fix $K\in \Nat$, and observe
		\[ \sum_{k:1}^K f(x_k) = \sum_{k:1}^K \lim_n f_n(x_k) = \lim_n \sum_{k:1}^K f_n(x_k) \leq
		\lim_n \sum_{k:1}^\infty f_n(x_k) = L\]
		
		Letting $K\to\infty$ gives  $\sum_{k:1}^\infty f(x_k) \leq L$
		
\end{proof}

}

Hence, given $\mu_n: \Obs\to [0,1]$ and $\brho(\bU) = \lim_{n\to \infty} \mu_n(\bU)$, the following holds:
{ $$\lim_{n\to \infty}\sum_{\bU\in \Obs} \mu_n(\bU) ~=~  \sum_{\bU\in \Obs} \brho(\bU)$$}

\condinc{}{
Recall:
\begin{itemize}
	
	\item If $\m,\s$ are multi-distributions (Sec. \ref{sec:multi} ), $\mu, \sigma$ denote the associated distribution.
	
	\item Greek bold letters denote limit distributions (Def.\ref{def:limits} )
	
	\item Given a distribution $\mu$ on a countable set $\mathcal X$,  we recall the definition of norm $\norm \mu= \sum_{x\in \mathcal X} \mu(x)$.
\end{itemize}
}

\vskip 4pt
\paragraph{Existence of maximals.}

We recall the definition of norm $\norm \mu= \sum_{x\in \mathcal X} \mu(x)$.
\newcommand{\Norms}{\texttt{Norms}}
\begin{lemma}[Existence of maximals] Confluence implies that:
	
	\begin{enumerate}
		\item  $\Norms(\m) = \{ \norm {\bmu} ~\mid~  \bmu\in Lim(\m) \}$ has a greatest element;
		\item $Lim(\m)$ has maximal elements.
	\end{enumerate} 
\end{lemma}

\begin{figure}[h]\centering
	\fbox{		\includegraphics[width=0.4\textwidth]{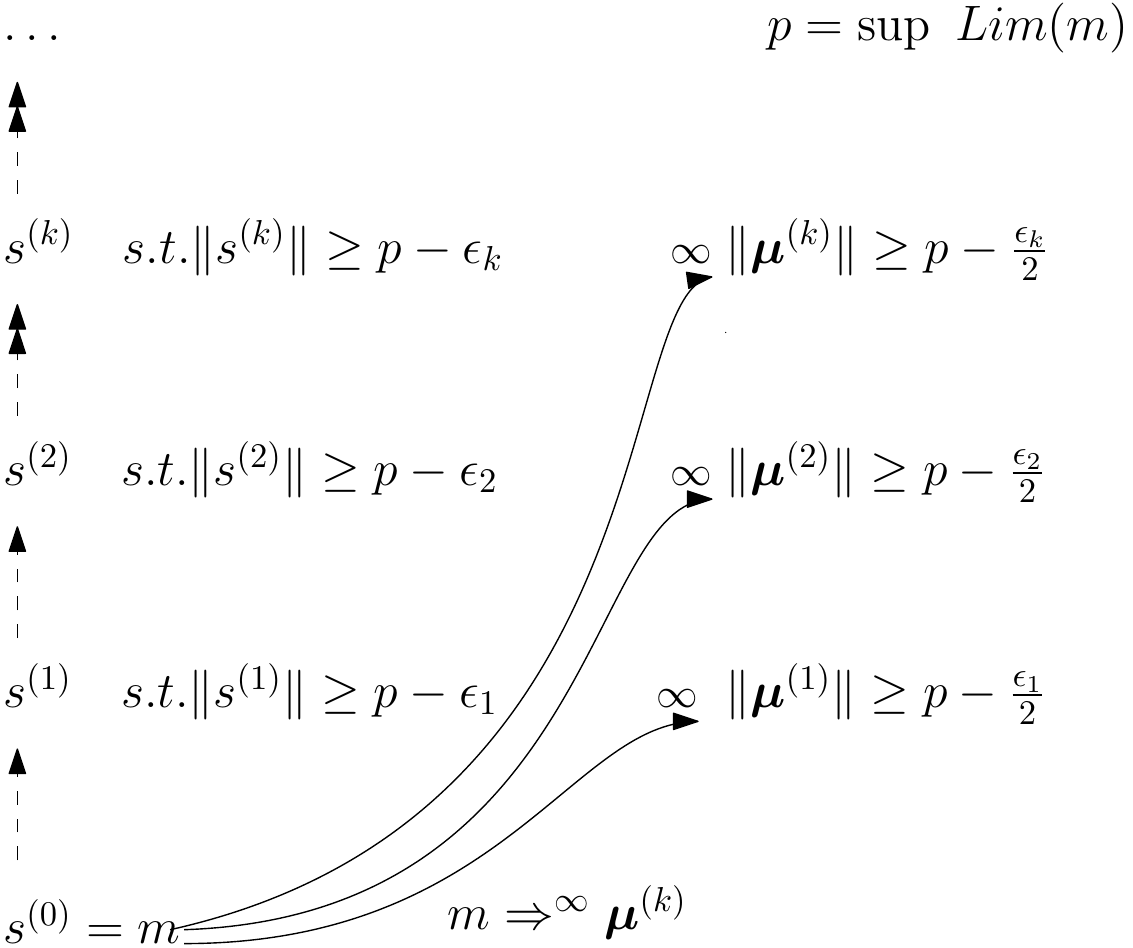}}
	\caption{A sequence whose limit distribution is a maximal element of $\Lim(\m)$}\label{fig:exists}
\end{figure}	
\begin{proof}
	
	\textbf{(1. )}
	Let $p =\sup~{\Norms(\m)}$. 
	We show that   $p\in \Norms(\m)$, 
	by providing  a rewrite  sequence $\seq \m$ from $\m$ such that  $\seq \m \tolim \btau$ and $\norm \btau =p$.

	The following  facts are all easy to check:
	\begin{itemize}
		\item[a.] If $\alpha <\beta$ then  $\norm \alpha < \norm \beta$.
		
		\item[b.]\label{Lk}  If  $p\not\in \Norms(\m)$, then for each $\epsilon$, there exists   $\bmu\in Lim(\m) $ such that   $ \norm {\bmu} \geq p-\epsilon$. 
		
		\item[c.]\label{Lmain} The Main Lemma implies that, fixed $\epsilon$,  if $\m \tolim \bmu$ with $\norm \bmu \geq (p-\epsilon)$, and $\m \Red^* \s$, then  there exists $\s'$, such that $\s \Red^* \s' $ and $\norm{\sigma'} \geq ( p-2\epsilon)$.\\ (Proof: 
		Main Lemma implies that there is a rewrite sequence $\seq \s$ from $\s$ which converges to  $\bsigma\geq \bmu$.
		Therefore $ \seq \s \tolim \bsigma$ where $\norm \bsigma \geq (p-\epsilon)$. For the same $\epsilon$,  there is an index $N$   such that  $\s\Red^* \s_N$ and $\norm {\sigma_N}\geq (\norm \bsigma-\epsilon) $, hence $\norm {\sigma_N}\geq p-2\epsilon$. )
		
	\condinc{}{	\SC{Qui non capisco il passaggio $\norm {\sigma_N}\geq (\norm \bsigma-\epsilon) $. Io capisco il punto c. in questo modo:\\
			By point b.), there is $\bmu\in Lim(\m) $ such that   $ \norm {\bmu} \geq p-\epsilon$. So, by Main Lemma, for each rewrite sequence $\seq \s$ such that $\m \Red^* \s $,  $\bsigma\geq \bmu$. So, by point b., there is $\s_{2\epsilon}$ such that $\m \Red^* \s_{2 \epsilon} $ and
			$\norm \bsigmaep \geq (p - 2\epsilon)$.
		}	}	
		
		\item[d.] $\forall \delta\in \Real^+$ there exists $k$ such that $ \frac{p}{2^k} \leq \delta$.
	\end{itemize}


	For each $k\in \Nat$, let  $\epsilon_k=\frac{p}{2^k}$. Let $\s^{(0)}=\m$. From here, we build a sequence of reductions
	$\m\Red^* \s^{(1)}\Red^*s^{(2)}\Red^*\dots$ whose limit has norm $p$, as illustrated in Fig.~\ref{fig:exists}.
	For each $k>0$, we observe that:
	
	\begin{itemize}
		\item By (b.) there exists {$\bmu^{(k)}\in Lim(\m)$} such that   $ \norm {\bmu^{(k)}} \geq (p-\frac{1}{2} \frac{p}{2^k})$. 
		\item  From $\m\Red^* \s^{(k-1)}$, we use (c.) to establish that there exists  $\s^{(k)}$ such that $\s^{(k-1)}\Red^* \s^{(k)}$ and  
		$\norm{{\sigma^{(k)}}} \geq (p- \frac{p}{2^k})$. Observe that  $\bmu^{(k)},\s^{(k-1)},\s^{(k)} $ resp.  instantiate $\bmu, \s,\s'$ of (c.).
	\end{itemize}
	
	Let $\seq \s$ be the concatenation of all the  finite sequences $\s^{(k-1)}\Red^* \s^{(k)}$. By construction,  
	$\seq \s \tolim \btau$ such that  $\norm \btau=p$. Hence  $p\in  \Norms(\m)$.
	
	\textbf{(1. $\implies$ 2.) } We observe that if  $\seq \m \tolim \bmu $ and $\norm \bmu$ is maximal in $\Norms(\m)$, then $\bmu$ is maximal in $Lim(\m)$, because  of (a.).
	
\end{proof}

\subsection{Proofs of Section~\ref{sec:Plinear}}

\subsubsection{Confluence of $\PLambda^!$}\label{sec:app_conflin}
We prove Theorem~\ref{thm:conf_lin}. First, we need to prove some preliminary results.
\begin{lemma}\label{lem:comm_pointwiseLin}
If $M\redbl \n$ and $M\redo \s$, then exists $\r$ such that $\n\Redo \r$ and $\s\Redbl \r$ 
\end{lemma}
\begin{proof}
We reason by induction on $M$. The key case is case~\ref{case:key}. Case $M=x$ and $M=!P$ are not possible given the  hypothesis.

\begin{enumerate}
\item 
Case $M= P\oplus Q$. Similar to Lemma \ref{lem:comm}, case (\ref{case:sum})


 \item Case $M =\ss(Q)$, and both redexes are inside $Q$.  Similar to Lemma \ref{lem:comm}, case  (2.\ref{case:ind}).

\item Case $M=PQ$, with  the $\beta$-redex inside $P$, and the $\oplus$-redex inside $Q$. Similar to Lemma \ref{lem:comm},
case (2.\ref{case:disjoints}).

\item Case $M=(\lam !x. P)!Q$, where $M$ is the  $\beta$-redex.  The $\oplus$-redex needs to be inside $P$.
Assume $P\redo\mset{\two P_1, \two P_2}$.
We have $M \redo \mset{\two (\lam !x.P_1) !Q, \two (\lam !x. P_2) !Q}$, and 
$M \redbl  \mset{P[Q/x]}$. It is immediate that the multidistribution $\r=\mset{\two P_1[Q/x], \two P_2[Q/x]}$ satisfies the claim.

\item\label{case:key} Case $M=(\lam x. P)Q$, where $M$ is the  $\beta$-redex. 
If the $\oplus$-redex is inside $P$, we reason as above. Assume  that the $\oplus$-redex is inside $Q$, and we have $Q\redo\mset{\two Q_1, \two Q_2}$. The key observation is that in $P$ there is \textit{at most one occurrence }of $x$. Let  assume there is exactly one occurrence (the case of none is easy). Let  $\cc$ be the context such that $P=\cc(x)$ (\ie, $\cc$ is $P$, with a hole in the place of $x$). Observe that $P[Q/x]=\cc(Q)$. We have $M\redbl \mset{P[Q/x]=\cc(Q)}$, and $M\redo \mset{\two (\lam x. P)Q_1, \two (\lam x. P)Q_2}$.
The  multidistribution $\r=\mset{\two \cc(Q_1), \two \cc(Q_2)}$ satisfies the claim.

\end{enumerate}	
\end{proof}

\textbf{Lem.~\ref{lem:conf_lin}}:	
\begin{enumerate}
			\item 	The  reduction $\Redo$ is diamond.
			\item 	The reduction $\Redb$  is confluent.
						\item 	The reductions $\Redb$ and $\Redo$ commute.
		\end{enumerate}
	\begin{proof}
		\begin{enumerate}

			\item Same proof as for Lemma \ref{lem:oplusCR}.
			\item Inherited from $\Lambda^!$ via the translation  $\trb{.}$  and Prop.~\ref{prop:translationlin}.
				\item 
			
			We prove that $\Redb$ and $\Redo$ $\diamond$-commute, 
			by using  Lemma~\ref{lem:pointwise} and Lemma~\ref {lem:comm_pointwiseLin}.
			
%
		\end{enumerate}
	\end{proof}
	
\textbf{Thm.~\ref{thm:conf_lin}}
		 The reduction $\Red$ of $\PLambda^!$   is confluent.
	\begin{proof}
		By  Hindley-Rosen Lemma, from Lemma~\ref{lem:conf_lin}.
	\end{proof}


\subsubsection{Surface Standardization in $\PLambda^!$}\label{sec:postponement_lin}
In order to prove Proposition \ref{thm:surfacestandard_lin}, first we prove a lemma.
\begin{lemma}\label{lem:post_lin}If $M\dparred M'$ and $M' \sred \n$, then $\mset{M} \sRed^* \s$ and $\s \dparRed \n$.
\end{lemma}

\begin{proof}If $M' \sredbl \n$, the claim holds by \emph{simulation} in $\Lambda^!$.
	If $M' \sred_{\oplus} \n$, we procede by induction on $M$. 
	\begin{enumerate}
		\item The case $M=x$ and $M= !P$ do not apply.
		\item Let $M=P\oplus Q$. Similar to Lemma~\ref{lem:post}, Point (2.).
		\item 	 Let $M=PQ,~ \lam !x.P$ or $\lam x.P$. Similar to Lemma~\ref{lem:post}, Point (3.).
	\end{enumerate}
	
\end{proof}

\begin{cor}\label{cor:post_lin}
	If $\m\dparRed \n$ and $\n \sRed^* \r$, then exists $\s$ with $\m \sRed^* \s$ and $\s \dparRed \r$.
\end{cor}
\begin{proof}
	Same as Lemma~\ref{cor:post}, using Lemma~\ref{lem:post_lin}.
\end{proof}

Then we can prove: 

\textbf{Prop. \ref{thm:surfacestandard_lin}}\
	If $\m\Red^*\n$ then  there exists $\r$ such that  $\m \sRed^* \r$ and $\r \dRed^* \n$. 

\begin{proof}
	Same as the proof of Thm.~\ref{thm:surfacestandard}, by induction on the length of the reduction $\m\Red^*\n$,  using  this time Corollary \ref{cor:post_lin}. 
	
	%
\end{proof}


\subsection{Proofs of Section~\ref{sec:CBN}}
\subsubsection{ $\PLambda^!$ is a conservative extension of $\PLambda^\cbn$ (\ref{sec:tr_cbn})}

To prove Proposition~\ref{thm:ntr}, we first prove that $\trn{\cdot}$ preserves both surface contexts and  $\oplus$-redexes.
\begin{lemma}\label{lem:preserveS} Given $M\in \PLambda^\cbn$, (S1) holds $\Longleftrightarrow$ (S2) holds, where:
	\begin{enumerate}
		\item[S1:] in $\PLambda^\cbn$ , there exists $\ss$ surface context and a redex $r=R_1\oplus R_2$ such that $M=\ss(r)$;
		\item[S2:] in $\PLambda^!$ there exists $\tc$ surface context and a redex $u=U_1\oplus U_2$ such that $\trn{M}=\tc(u)$;
	\end{enumerate}	
	and moreover  $\trn{\ss(R_i)}=\tc(U_i)$, for $i\in \{1,2\}$.
	
\end{lemma}

\begin{proof}
	$  \Longrightarrow $. By induction on the form of the surface context.
	\begin{itemize}
		\item  $\square$. Since $M=r$, then  $\trn{M}=\trn{R_1}\oplus \trn{R_2}$. Hence $u=\trn{R_1}\oplus \trn{R_2}$ and $\tc= \square$
		satisfy the claim.

		\item $\ss Q$. We have that $M=\ss Q(r)=\ss(r)Q$. Hence  $ \trn{\ss(r) Q} = \trn{\ss(r)}!\trn{Q}$. By inductive hypothesis, there exist $\tc'$ and $u$ such that 
		$ \trn{\ss(r)}=\tc'(u) $, and   $\trn{\ss(R_i)}=\tc'(U_i)$.
		By definition of surface context in $\PLambda^!$, the claim hold with $\tc=\tc'!\trn{Q}$, and the same $u$.

		\item $\lam x.\ss$.  $\trn{ \lam x. \ss (r)} = \lam ! x. \trn{ \ss' (r)}$, and the claim holds by inductive hypothesis.
	\end{itemize}

	$ \Longleftarrow $.
	We examine the possible form of $\tc$, given that $\trn{M}=\tc(u)$; we prove that 
	$M=\ss(r)$ and that $\trn{\ss(R_i)}=\tc(U_i)$.
	\begin{itemize}
		\item  $\square$. Immediate.
		
		\item $\tc' Q$. We have that $(\tc' Q)(u)=\tc'(u)Q$, and
		$\tc'(u)Q=\trn{LN}=\trn{L} !\trn{N}$ with $M=LN$. Therefore  $\tc'(u)=\trn{L}$,  and the claim holds by inductive hypothesis and definition of surface context. 
		
		\item $\lam !x.\tc'$.  We have that $(\lam !x.\tc')(u)=\lam !x.\tc'(u)$ and  $\lam !x.\tc'(u)= \lam !x.\trn{M'}$, with $M=\lam x. M'$. The claim holds by inductive hypothesis. 
		
	\end{itemize}

\end{proof}

\begin{proposition*}\textbf{\ref{thm:ntr}}. [Simulation] The  translation $\trn{.}$ is sound and complete; it preserves surface reduction and surface normal forms.
		Let $M\in \PLambda^\cbn$; the following hold:
		\begin{enumerate}
			\item if $M \red \n$  then $\trn{M}\red\trn \n$;
			\item if $M \sred \n$  then $\trn{M}\sred\trn \n$;
			
			\item if $\trn{M}\red \s$ then $\exists ! \n$  such that $\s= \trn \n$ and $M\red \n$;
			\item if $\trn{M}\sred \s$ then $\exists ! \n$  such that $\s= \trn \n$ and $M\sred \n$;
			
			\item $M\in \Hnf$ if and only if $\trn{M} \in \Snf^!$.
	\end{enumerate} 	
\end{proposition*}

\begin{proof}

We prove (1.)-(4.); since
$\red~=~\red_{\beta}  \cup \red_{\oplus} $, we deal separately with the two reductions. Point (5.) is an immediate consequence of the other points.

%
%
\begin{itemize}
	\item[$\redbn$]We deal with $\red_{\beta}$  via simulation in $\Lambda^\cbn$ and $\Lambda^!$, since the analogous   result is proved in  \cite{gm18}.
We have defined a  translation $(-)_{!}:\PLambda^!\to \Lambda^!$ which is sound and complete, and preserves surface reduction. It is straightforward to define a similar translation from $\PLambda^\cbn$ into $\Lambda^\cbn$. 
Therefore, if 
in $\PLambda^\cbn$ it holds $M \red_{\beta} N$, 
we translate in $\Lambda^\cbn$, 
use the  result in \cite{gm18} and conclude (via simulation) that
 $\trn{M}\red_{\beta}\trn N$ in $\PLambda^!$. 
 Similarly for (2.)-(3.)-(4.). 
 
 \item[$\redo$] Immediate consequence of Lemma~\ref{lem:preserveS}, 
 which proves that  that $\trn{\cdot}$ preserves both surface contexts and  $\oplus$-redexes.
\end{itemize}
\end{proof}

\condinc{}{
	\begin{proposition*}[\ref{thm:ntr}. Simulation] The  translation $\trn{.}$ is sound and complete; it preserves surface reduction and surface normal forms.
		Let $M\in \PLambda^\cbn$; the following hold:
		\begin{enumerate}
			\item if $M \red \n$  then $\trn{M}\red\trn \n$;
			\item if $M \sred \n$  then $\trn{M}\sred\trn \n$;
			
			\item if $\trn{M}\red \s$ then $\exists ! \n$  such that $\s= \trn \n$ and $M\red \n$;
			\item if $\trn{M}\sred \s$ then $\exists ! \n$  such that $\s= \trn \n$ and $M\sred \n$;
			
			\item $M\in \Hnf$ if and only if $\trn{M} \in \Snf^!$.
		\end{enumerate} 	
	\end{proposition*}

	\begin{proof}
		(5.) is an immediate consequence of the previous points. We prove (1.)-(4.).
		$\red_{\cbn}=\red_{\cbn\beta}  \cup \red_{\cbn\oplus} $ and $\red_{!}=\red_{!\beta } \cup \red_{!\oplus}$. We deal with beta reductions via simulation, since the corresponding results are proved in  \cite{gm18}.
		We have defined a  translation $(-)_{!}:\PLambda^!\to \Lambda^!$ which is sound and complete, and preserves surface reduction. It is straightforward to define a similar translation from $\PLambda^\cbn$ into $\Lambda^\cbn$. 
		If $M \red_{\cbn\beta} N$, we translate, 
		use the  result in \cite{gm18} and conclude that $\trn{M}\red_{!\beta}\trn N$. Similarly for (2.)-(4.). We prove (1.)-(4.) for $\redo$.
		
		\begin{itemize}
			\item (1.) and (2.). We prove that if $M \sred_{\cbn\oplus} \m=\mset{M_1,M_2}$  then $\trn{M}\red_{!\oplus}\trn \m= \mset{\trn{M_1},\trn{M_2}}$, by induction on $M$. We observe that $M$ has the form $\ss(R)$.
			\begin{enumerate}
				\item $\square$.  $M=M_1\oplus M_2$, and $\trn{M}=\trn{M_1}\oplus \trn{M_2}$, hence  the claim holds.
				\item $\lam x.\ss$.  $M = \lam x. N$ and $N\sred_{\cbn\oplus} \mset{N_1,N_2}$. By i.h., $\trn{N}\sred_{\cbn\oplus} \mset{\trn{N_1},\trn{N_2}}$, hence the claim is immediate.
				\item $\ss U$. $M=NU$,   and $N\sred_{\cbn\oplus} \mset{N_1,N_2}$. 
				By i.h. $\trn{N}\sred_{!\oplus} \mset{\trn{N_1},\trn{N_2}}$.
				Therefore $\trn{M}=\trn{N} !\trn{U}$  
				$\sred_{!\oplus}$ $\mset{\trn{N_1}\trn{U},\trn{N_2}\trn{U}}=\trn{\m}$.
				
			\end{enumerate}
			\item  (3.) and (4.). We prove that if $\trn M\sred_{!\oplus} \s=\mset{S_1,S_2}$  then exists  $\trn \n= \s=
			\mset{\trn{ N_1},\trn{N_2}}$
			and $M\sred_{\cbn\oplus} \mset{S_1,S_2}$,  by induction on $M$. We only examine the case $M=PQ$, the others are similar.
			
			$\trn M=\trn{(AB)}=\trn{A}!\trn{B}$. The redex of $\trn M$ cannot be in $\trn B$. Therefore, $\trn{A}\sred_{!\oplus}\mset{P_1,P_2}$, and
			$S_i=P_i !\trn B$. By i.h., exists $L_i\in \PLambda$, $\trn{L_i}=P_i$ and $A\sred \mset{L_i}$. We conclude that $M=AB\sred \mset{
				L_iB  }$ and $\trn{(L_iB)}=S_i$.
		\end{itemize}
	\end{proof}}
}
\end{document}

%% file: macros_long.tex
\newcommand{\SC}[1]{\textcolor{magenta}{Simona: #1}}
\newcommand{\CF}[1]{\textcolor{violet}{*CF: #1 *}}

\usepackage[T1]{fontenc} 
\usepackage{amsmath}
\usepackage{url}
\usepackage{amsthm}
\usepackage{color}
\usepackage{graphicx}

\usepackage{cmll}
\usepackage{stmaryrd}
\usepackage{proof}
\usepackage{wrapfig}
\usepackage{amssymb}
\usepackage{bm} 

\usepackage[english]{babel}
\newcommand{\ccontext}{{\bf C}}
\newcommand{\cc}{\ccontext}

\newcommand{\hh}{{\bf H}}
\newcommand{\kk}{{\bf K}}

\newcommand{\lc}{{\bf L}}
\renewcommand{\ss}{\bm S}
\newcommand{\tc}{\bm T}
	
\newcommand{\seq}[1]{\langle#1_n\rangle_{n\in\Nat}}

\newcommand{\equbetav}{=_{\beta_{v}}}

\newcommand{\xor}{\mathtt{~XOR~}}

\newcommand{\cbv}{{\mathtt{cbv}}}
\newcommand{\cbn}{{\mathtt{cbn}}}

\newcommand{\red}{\rightarrow}
\newcommand{\redb}{\rightarrow_{\beta}}
\newcommand{\redbv}{\rightarrow_{\beta_v}}
\newcommand{\redo}{\rightarrow_{\oplus}}
\newcommand{\full}{\rightrightarrows}

\newcommand{\Red}{\Rightarrow} 
\newcommand{\Redb}{\Red_{\beta} }  
\newcommand{\Redbv}{\Red_{\beta_v} }  
\newcommand{\Redo}{\Red_{\oplus}}

\newcommand{\surf}{\textsf s}
\newcommand{\deep}{\textsf d}
\newcommand{\lleft}{\textsf l}
\renewcommand{\int}{\textsf {int}}
\newcommand{\head}{\textsf h}

\newcommand{\ired}{\overset{{{}_{\int}}}{\red}}
\newcommand{\iRed}  {\overset{{}_{\int}}{\Red}}

\newcommand{\iredbv}{\ired_{\beta_v}}

\newcommand{\lred}{\overset{{}_\lleft}{\red}}
\newcommand{\lRed}{\overset{{}_\lleft}{\Red}}

\newcommand{\lredbv}{\lred_{\beta_v}}

\newcommand{\dred}{\overset{{{}_{\deep}}}{\red}}
\newcommand{\dRed}  {\overset{{}_{\deep}}{\Red}} 
\newcommand{\dredbv}{\dred_{\beta_v}}

\newcommand{\sred}{\overset{{}_\surf}{\red}}
\newcommand{\sRed}{\overset{{}_\surf}{\Red}}
\newcommand{\sredb}{\sred_{\beta}}
\newcommand{\sredbv}{\sred_{\beta_v}}
\newcommand{\sredo}{\sred_{\oplus}}

\newcommand{\sRedbv}{\sRed_{\beta_v}}

\newcommand{\sredbl}{\sred_{\beta}}
\newcommand{\redbl}{\red_{\beta}}
\newcommand{\Redbl}{\Red_{\beta}}

\newcommand{\redbn}{\red_{\beta}}

\newcommand{\hred}{\overset{{}_\head}{\red}}

\newcommand{\lfull}{\oset{\lleft}{\full}}
\newcommand{\sfull}{\oset{\surf}{\full}}
\newcommand{\hfull}{\oset{\head}{\full}}

\newcommand{\sfullbv}{\sfull_{\beta_v}}

\newcommand{\parmark}{\shortparallel\mkern -3mu}

\newcommand{\makepar}[1]{~\parmark \mkern-8mu #1}

\newcommand{\parred}{{\makepar \red}}
\newcommand{\parRed}  {{\makepar \Red} }
\newcommand{\iparred}{{\makepar \ired}}

\newcommand{\parredbv}{\parred_{\beta_v}}

\newcommand{\dparred}{{\makepar \dred}}
\newcommand{\dparRed}  {{\makepar \dRed}}

\newcommand{\dparredbv}{\dparred_{\beta_v}}

\newcommand{\PLambda}{\Lambda_\oplus}
\newcommand{\trn}[1]{(#1)_{{}_\mathtt{N}}}
\newcommand{\trb}[1]{(#1)_{!}}

\newcommand{\true}{\mathtt{T}}
\newcommand{\false}{\mathtt{F}}

\newcommand{\sing}{{{}_{\{\}}}}

\newcommand{\Val}{\mathcal V}
\newcommand{\Snf}{\mathcal S}
\newcommand{\Nnf}{\mathcal N}
\newcommand{\Hnf}{\mathcal H}

\newcommand{\st}{\mid}

\renewcommand{\S}{\mathfrak s}
\newcommand{\T}{\mathfrak t}
\renewcommand{\L}{\mathfrak l}

\newcommand{\two}{\frac{1}{2}}

\newcommand{\Real}{\mathbb{R}}
\newcommand{\Nat}{\mathbb{N}}

\renewcommand{\implies}{\Rightarrow}

\theoremstyle{plain}
\newtheorem{theorem}{Theorem}
\newtheorem{lemma}[theorem]{Lemma}
\newtheorem{fact}[theorem]{Fact}
\newtheorem{example}[theorem]{Example}
\newtheorem{remark}[theorem]{Remark}

\newtheorem{thm}[theorem]{Theorem}
\newtheorem{prop}[theorem]{Proposition}
\newtheorem{cor}[theorem]{Corollary}

\theoremstyle{definition}
\newtheorem{Def}[theorem]{Definition}

\newtheorem*{theorem*}{Theorem}
\newtheorem*{proposition*}{Proposition}
\newtheorem*{lemma*}{Lemma}
\newtheorem*{remark*}{Remark}
\newtheorem*{notation*}{Notation}
\newtheorem*{fact*}{Fact}

\newcommand{\lam}{\lambda}
\newcommand{\tr}[1]{(#1)_{\lambda}}

\newcommand{\ie}{\emph{i.e.}}

\usepackage{xcolor}

\newcommand{\blue}[1]{{\color{blue}{#1}}}
\newcommand{\green}[1]{{\color{green}{#1}}}

\renewcommand{\tt}{\mathbf{t}}

\newcommand{\Lim}{\mathtt{Lim}}

\newcommand{\down}{\Downarrow}

\newcommand{\tolim}{\overset{{}_{~~~\infty}}{\Rightarrow}}
\newcommand{\tolimx}[1]{\overset{{}_{\infty}}{\Rightarrow_{\mkern-4mu{}_{#1}}}}
\newcommand{\tolims}{\overset{{}_{~~\surf~\infty }}{\Rightarrow}}
\newcommand{\toliml}{\overset{{}_{~~\lleft~\infty }}{\Rightarrow}}

\newcommand{\DST}[1]{\mathtt{DST}(#1)}

\newcommand{\MDST}[1]{\mathtt{MDST}(#1)}

\newcommand{\supp}[1]{\mathit{Supp}(#1)}

\newcommand{\Obs}{\mathtt{Obs}}
\newcommand{\A}{\mathcal{A}}

\newcommand{\m}{\mathtt m}   
\newcommand{\n}{\mathtt n} 
\renewcommand{\r}{\mathtt r}     
\newcommand{\s}{\mathtt s}

\newcommand{\den}[1]{\llbracket {#1} \rrbracket}

\newcommand{\fullo}{\full_{\oplus}}

\newcommand{\bV}{\mathbf V}

\newcommand{\bU}{\mathbf U}
\newcommand{\bZ}{\mathbf Z}

\newcommand{\bmu}{\pmb \mu}

\newcommand{\bsigma}{\pmb \sigma}
\newcommand{\btau}{\pmb \tau}
\newcommand{\brho}{\pmb \rho}
\newcommand{\bsigmaep}{\pmb \sigma_{2\epsilon}}

\newcommand{\norm}[1]{\|#1\|}

\newcommand{\set}[1]{\{#1\}}
\newcommand{\multiset}[1]{\pmb[ #1 \pmb]}

\newcommand{\mset}{\multiset}
\newcommand{\mdist}[1]{\mset{#1}}

\newcommand{\iI}{i \in I}
\newcommand{\jJ}{j \in J}

\makeatletter
\newcommand{\oset}[3][0ex]{%
	\mathrel{\mathop{#3}\limits^{
			\vbox to#1{\kern0\ex@
				\hbox{$\tiny#2$}\vss}}}}
\makeatother

\usepackage{floatflt}
\usepackage[font=scriptsize,skip=4pt]{caption}

\newcounter{number}